\documentclass[11pt,a4paper]{article}
\usepackage[utf8]{inputenc}
\usepackage[T1]{fontenc}
\usepackage{lmodern}
\usepackage[margin=1in]{geometry}
\usepackage{amsmath,amssymb,amsthm,mathtools}
\usepackage{booktabs}
\usepackage{array}
\usepackage{enumitem}
\usepackage{graphicx}
\usepackage{float}
\usepackage[
  colorlinks=true,
  linkcolor=blue,
  citecolor=blue,
  urlcolor=blue,
  pdftitle={BPS spectra of Tr[Psi\string^p] matrix models for odd p},
  pdfauthor={Miguel Tierz}
]{hyperref}
\usepackage{cite}

\newcommand{\tr}{\mathrm{Tr}}
\newcommand{\Zbps}{Z_{\mathrm{BPS}}}
\newcommand{\SU}{\mathrm{SU}}
\newcommand{\C}{\mathbb{C}}
\newcommand{\Z}{\mathbb{Z}}

\newtheorem{theorem}{Theorem}
\newtheorem{proposition}{Proposition}
\newtheorem{corollary}[proposition]{Corollary}
\newtheorem{conjecture}{Conjecture}
\theoremstyle{definition}
\newtheorem{remark}{Remark}
\newtheorem{numericalresult}{Numerical result}

\begin{document}

\title{BPS spectra of $\tr[\Psi^p]$ matrix models for odd $p$}
\author{Miguel Tierz\\[4pt]
\small Shanghai Institute for Mathematics and Interdisciplinary Sciences (SIMIS),\\
\small Shanghai 200433, China\\[2pt]
\small \texttt{tierz@simis.cn}}
\date{}
\maketitle

\begin{abstract}
We study the BPS cohomology and Hamiltonian of a model built from an
\(N\times N\) matrix \(\Psi\) of complex fermions, with supercharge
\(Q_p=\tr(\Psi^p)\) for odd \(p\ge3\).  Numerical calculations give the
BPS multiplicity at every fermion number for
\((p,N)=(5,3),(5,4),(5,5),(7,4)\), and through fermion number six at
\((7,5)\).  In \(\Zbps^{(p,N)}(x)=\sum_Rh_Rx^R\), \(x\) records fermion
number and \(h_R\) counts BPS states in sector \(R\).  In each of the four
cases computed at every fermion number, factoring out the lowest power
of \(x\) leaves a polynomial divisible by a positive power of \(p\) and
by \((1+x)^N\).  The quotient by these factors has nonnegative integer
coefficients.  Exact
decoupling of \(\tr\Psi\) and the Euler
characteristic prove divisibility by \((1+x)^2\).  Pairing sectors of
complementary fermion number gives a third factor when \(N\) is odd.  For
odd \(p\le2N-1\), the only range in which \(Q_p\) can be nonzero, the same
pairing proves the corresponding rank symmetry for \(Q_p\).  The additional factors
required to reach \((1+x)^N\) remain conjectural in general.
Conditional on reported numerical vanishings at larger \(N\), the
projections identify as fortuitous those BPS classes with no preimage
from some larger matrix size.  Evaluating at \(x=1\) gives the total BPS
count \(\Zbps^{(p,N)}(1)\).  For each fixed odd \(p\), the Witten indices
imply that, for any sequence of matrix sizes \(N_j\to\infty\) along which
\(N_j^{-2}\log\Zbps^{(p,N_j)}(1)\) converges, its limit lies in
\(\bigl[\log(2\cos\frac{\pi}{2p}),\log 2\bigr]\).  After dividing the Hamiltonian
by \(p^2\), normal ordering gives an alternating sum of operators \(K_k\),
with \(K_k\) containing \(k\) creation and \(k\) annihilation operators.
We obtain exact formulas for \(K_0,K_1,K_2\) when \(p=5\) and \(N\ge3\).  Only
\(K_3,K_4\) can contain spectral information beyond the number of
traceless fermions and the quadratic Casimir.  At \(N=3\), the Hodge star
maps the invariant cubic form to the quintic form up to scale, so all
five terms commute.  Calculations with integer matrices give nonzero
commutators in the reported \(N=4\) sectors.
\end{abstract}

\section{Introduction}

We study the quantum mechanics of a single \(N\times N\) matrix \(\Psi\)
whose entries \(\Psi_{ij}\) are fermion creation operators.  The group
\(\SU(N)\) is a global symmetry and is not gauged, so the full Hilbert
space is retained.  The supercharge and Hamiltonian are
\begin{equation}\label{eq:Qp}
Q_p=\tr(\Psi^p),\qquad p\ge 3\text{ odd},
\end{equation}
where \(\{\Psi_{ij},\Psi_{kl}^\dagger\}=\delta_{ik}\delta_{jl}\) and
\(H=\{Q_p,Q_p^\dagger\}\)~\cite{Witten1982,Troost2020}.  The Hilbert
space is the fermionic Fock space
\(\mathcal H=\Lambda^\bullet(\C^{N\times N})\), whose component of degree
\(R\) has fermion number \(R\).  For odd \(p\), \(Q_p^2=0\).  Since
\(H=\{Q_p,Q_p^\dagger\}\), the BPS space is
\(\ker Q_p/\operatorname{im}Q_p\): each cohomology class has a
representative of zero energy, and every state of zero energy defines
such a class.

We determine the BPS multiplicities at finite \(N\), compare BPS classes
under projections between ranks, and identify the terms that can contain
spectral information not fixed by the quadratic Casimir.

For \(p=3\), Chen obtained the exact total generating function
\cite{Chen2025}.  In this case the Hamiltonian is a constant minus the
quadratic Casimir, in accordance with Kostant's analysis of the
Laplacian and Casimir spectrum on an exterior algebra
\cite{Kostant1965}.  Reference~\cite{Chen2025} also noted the extensions
\(Q_p=\tr(\Psi^p)\) for odd \(p\), but did not study \(p\ge5\) in
detail.  For \(p\ge5\), the Hamiltonian is not in general determined by
the central Casimirs.  The numerical example at \((p,N)=(5,4)\) gives
two equivalent \(\SU(4)\) representations at different energies.

After normal ordering, the reduced Hamiltonian is
\(\widetilde H:=H/p^2=\sum_{k=0}^{p-1}(-1)^kK_k\), where \(K_k\)
contains \(k\) creation and \(k\) annihilation operators.  For \(p=5\), we determine
\(K_0,K_1,K_2\) for every \(N\ge3\).  Their alternating sum depends only
on the number of traceless fermions and the quadratic Casimir.  Only
\(K_3\) and \(K_4\) can contribute additional spectral information not
fixed by these operators.  In the numerical \((5,4)\), \(R=4\) example,
the splitting of the two equivalent 20-dimensional blocks in
\(\widetilde H\) is produced by \(-K_3+K_4\).  At \(N=3\), the traceless space for one particle has
dimension eight.  The Hodge star on its exterior algebra sends the
invariant cubic form to a nonzero multiple of the invariant quintic form,
making all five \(K_k\) commute.  Calculations with integer matrices show
that they do not all commute at \(N=4\), but this does not decide whether
the model is chaotic or integrable.

The present model is not intended as a holographic theory.  It provides
a finite dimensional system in which the BPS cohomology can be computed
and compared as \(N\) changes.  Within the cohomological analysis, we
keep two diagnostics separate.  The first is the difference between the
BPS count and the lower bound from the Witten index in each class of
fermion numbers modulo \(p\).  The second is whether a class belongs to
the stable image of the projections from larger \(N\).  Neither property
implies the other.

To formulate the second diagnostic, we adapt the distinction introduced
in Ref.~\cite{Chang2024}.  Here a class is monotone if it lies in the
stable image, the intersection of the images on cohomology induced by
projections from all larger matrices.  A class outside it is fortuitous.
That reference uses a covering complex, and we do not assume the two
constructions are equivalent.  This distinction has been studied in
\(\mathcal N=4\) supersymmetric Yang--Mills theory and in supersymmetric
Sachdev--Ye--Kitaev models
\cite{Chang2024,ChangChenSiaYang2024,ChoiChoiKimLeeLee2024}.  Related
constructions appear in the D1--D5 conformal field theory
\cite{ChangLinZhang2025,ChangZhang2025} and in
Aharony--Bergman--Jafferis--Maldacena theory
\cite{BelinABJM2025,BehanDeGioia2025}.

Further examples of fortuity have been studied in a model with higher
spin symmetry~\cite{KimLeeOh2025} and in a supergravity analysis of the
D1--D5 system~\cite{HughesShigemori2025}.  A relevant deformation of
\(\mathcal N=4\) Yang--Mills theory can produce fortuitous
cohomology~\cite{ChoiKim2025}.  Quantum corrections can also pair and
lift additional monotone and fortuitous classes~\cite{ChoiLee2025}.

At finite \(N\), trace relations in coupled matrix and vector systems
suggest a bosonic analogue of fortuity~\cite{deMelloKoch2025}.
Secondary invariants in invariant rings at finite \(N\) have also been
studied as algebraic models of nonperturbative
states~\cite{deMelloKochRodrigues2026}.  At \(N=3\), fortuitous black
hole states at rank \(N-1\) were combined with dual giant gravitons to
construct BPS cohomology at rank \(N\)~\cite{deMelloKochKim2025}.  At
tree level, the cohomology of \(Q\) can differ between gauge theories
related by \(S\) duality, with the reported mismatch involving one
fortuitous class and one monotone class~\cite{ChangLin2025Sduality}.

The calculations closest in method construct the supercharge complex
separately at each charge and determine its cohomology for small values
of \(N\)
\cite{Chang2024,ChangChenSiaYang2024,ChoiChoiKimLeeLee2024,
ChangLinZhang2025,BelinABJM2025}.  A related comparison for the dynamics
is the model of Ref.~\cite{BiggsMaldacena2026}, in which a fixed
supercharge replaces random couplings.  Its agreement with
supersymmetric Sachdev--Ye--Kitaev theory concerns the melonic
approximation for \(\SU(2)\) singlet observables at small charge.  The
model studied here has an exact \(\SU(N)\) symmetry that acts by
conjugation on the matrix fermions.  No melonic limit is constructed
here.

Other fermionic matrix models are discussed in
Refs.~\cite{AnninosDenefMonten2016,AnninosSilva,Tierz2017,
KlebanovMilekhin,KlebanovTASI,GaitanKlebanov2020,
SemenoffSzabo1996,PaniakSzabo2001}, including the single matrix model
of Ref.~\cite{BerensteinDeMelloKoch2019}.  A different family of
supersymmetric matrix models, motivated by the BMN model and by
holography in plane wave backgrounds, was proposed in
Ref.~\cite{Lee2026BMN}.  The invariant antisymmetric tensors used in the
supercharges come from Lie algebra cohomology
\cite{ChevalleyEilenberg,Koszul1950}.

We report numerical BPS multiplicities in every fermion number sector
for \((p,N)=(5,3)\), \((5,4)\), \((5,5)\), and \((7,4)\).  At \((7,5)\), only the sectors
\(R\le6\) are computed.  Section~\ref{sec:method} explains the numerical
rank assignments and their modular checks.  Unless an analytic proof is
given, the resulting multiplicities are numerical rather than certified
ranks over \(\mathbf Q\).

Let \(h_R\) be the BPS multiplicity at fermion number \(R\), and define
\(\Zbps^{(p,N)}(x)=\sum_R h_Rx^R\), with \(x\) recording fermion number.
In each of the four cases computed at every fermion number, the numerical
polynomial has the form
\[
\Zbps^{(p,N)}(x) = p^{b_{p,N}}\,x^{q_{\min}}\,(1+x)^N\,T_N^{(p)}(x).
\]
Here \(q_{\min}\) is the smallest fermion number carrying BPS states,
\(b_{p,N}\) is the exponent of the common positive power
\(p^{b_{p,N}}\), and
\(T_N^{(p)}\) has nonnegative coefficients.  Exact trace mode
decoupling and the Euler characteristic give at least two factors of \(1+x\).
Pairing complementary sectors in the highest exterior degree gives a
third when \(N\) is odd.  The factors still needed to reach \((1+x)^N\)
are conjectural.

Write \(r_R=\operatorname{rank}Q_R\), where \(Q_R\) maps
the sector with fermion number \(R\) to that with fermion number \(R+p\).
For odd \(p\le2N-1\), the same pairing gives
\(r_R=r_{N^2-p-R}\).  With
\(\mathcal R_{p,N}(x)=\sum_R r_Rx^R\), the exact
relation \((1+x)\mid\mathcal R_{5,5}(x)\) then determines the middle map
rank at \((5,5)\) from the ten reported ranks below it.

The projections identify all BPS classes at \((5,3)\), and the classes
in the sectors listed in Table~\ref{tab:fortuity} at \((5,4)\) and
\((7,4)\), as fortuitous.  These conclusions are conditional on the
reported numerical vanishings at larger \(N\).

Independently, set \(x=1\), so that \(\Zbps^{(p,N)}(1)\) is the total BPS
count.  Fix an odd \(p\) and let \(N_j\to\infty\) be any sequence of
matrix sizes for which
\(N_j^{-2}\log\Zbps^{(p,N_j)}(1)\) converges.  The Witten indices imply
that the resulting limit lies in
\[
\left[
\log\!\left(2\cos\frac{\pi}{2p}\right),\,\log 2
\right].
\]
Thus \(x=1\) and \(p\) are held fixed.  Only the matrix size tends to
infinity.

Figure~\ref{fig:regime} shows the finite pairs \((N,p)\) treated
numerically in this paper.

\begin{figure}[H]
\centering
\includegraphics[width=0.6\textwidth]{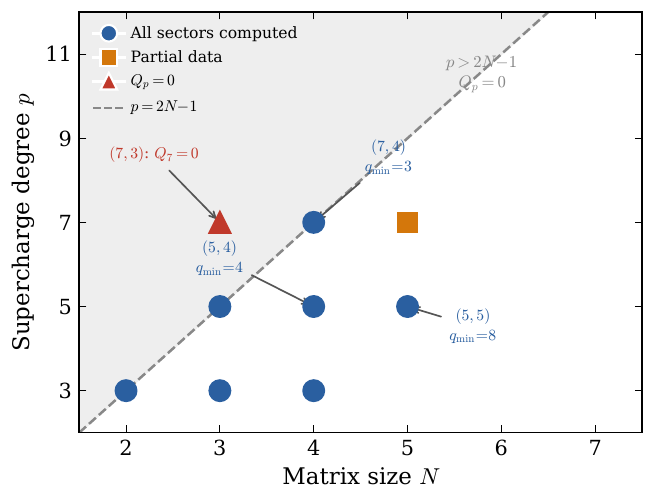}
\caption{Cases considered in this paper.
In the shaded region \(p>2N-1\), one has \(Q_p\equiv0\).
Blue circles mark cases with numerical BPS data in every sector.  The
orange square marks the partial data at \((7,5)\), and the red triangle
marks \((7,3)\), where \(Q_7=0\).
Annotations show the minimum BPS charge \(q_{\min}\) where known.}
\label{fig:regime}
\end{figure}

Section~\ref{sec:setup} defines the model and its cohomological
structure.  Section~\ref{sec:normalH} derives the Hamiltonian after
normal ordering and compares the commutators at \(N=3\) and \(N=4\).
Section~\ref{sec:method} describes the numerical method, presents the
BPS spectra, and compares different values of \(N\) by projection.
Section~\ref{sec:factorization} states the conjecture and proves the
exact identities for the ranks.  Section~\ref{sec:discussion} compares
the results with supersymmetric Sachdev--Ye--Kitaev models, derives the
bound at large \(N\), and gives the conclusions and open questions.  The
appendices contain deferred proofs, algebraic background, and the
sector data.

\section{Model and cohomological structure}
\label{sec:setup}

We describe the decomposition of the Fock space, the action of
\(\SU(N)\), the cohomological description, and the exact decoupling of
the trace fermion.  The Hamiltonian is studied separately in
Section~\ref{sec:normalH}.

\subsection{Fock space and sector decomposition}
\label{sec:fock}

The Hilbert space \(\mathcal H=\Lambda^\bullet(\C^{N\times N})\) is the
\(2^{N^2}\)-dimensional Fock space built from the \(N^2\) complex fermion
modes \(\Psi_{ij}\).
We use the convention that \(\Psi_{ij}\) acts by exterior multiplication and \(\Psi_{ij}^\dagger\) by contraction, so
\[
\Psi_{ij}^\dagger|0\rangle=0,\qquad
F:=\sum_{i,j}\Psi_{ij}\Psi_{ij}^\dagger ,
\]
where \(|0\rangle\) is the vacuum and \(F\) counts fermion number.
We separate the trace and traceless creation modes by
\begin{equation}\label{eq:trace-decomposition}
\chi:=\frac{1}{\sqrt N}\tr\Psi,\qquad
A:=\Psi-\frac{\chi}{\sqrt N}\mathbf 1,\qquad
\tr A=0 .
\end{equation}
They obey
\begin{equation}\label{eq:traceless-CAR}
\{\chi,\chi^\dagger\}=1,\qquad
\{\chi,A_{ij}^\dagger\}=0,\qquad
\{A_{ij},A_{kl}^\dagger\}
=\delta_{ik}\delta_{jl}-\frac{1}{N}\delta_{ij}\delta_{kl}.
\end{equation}
Thus the traceless space for one particle is
\(W_0=\{X\in\mathrm{Mat}_N(\C):\tr X=0\}\simeq\mathfrak{sl}_N(\C)\),
with \(\SU(N)\) acting by conjugation.
Throughout, \(\mathrm{Mat}_N(\C)\) carries the Hilbert--Schmidt inner
product \(\langle X,Y\rangle=\tr(X^\dagger Y)\), for which the matrix
units \(E_{ij}\) are orthonormal.  Here and below, ``orthonormal'' refers to this
inner product.
Choose an orthonormal Hermitian basis
\(\{T_\alpha\}_{\alpha=1}^{N^2-1}\) of \(W_0\) and write
\[
A=\sum_{\alpha=1}^{N^2-1}b_\alpha^\dagger T_\alpha,
\qquad
b_\alpha=(b_\alpha^\dagger)^\dagger .
\]
The operators \(b_\alpha^\dagger\) are the independent traceless
fermion creation modes.  They satisfy
\(\{b_\alpha,b_\beta^\dagger\}=\delta_{\alpha\beta}\).
All mixed anticommutators between the trace and traceless modes vanish.
In particular,
\(\{\chi,b_\alpha\}=\{\chi^\dagger,b_\alpha^\dagger\}=0\).
With \(T_0=\mathbf 1/\sqrt N\), the space for one particle therefore
decomposes orthogonally as
\(\mathrm{Mat}_N(\C)=\C T_0\oplus W_0\), and \(\chi\) is the
creation operator associated with \(T_0\).
In this convention the occupation number of the trace mode is
\(n_\chi:=\chi\chi^\dagger\), and
\(\hat n:=F-n_\chi\) counts the traceless fermions.
With \(V_R:=\Lambda^R(\C^{N\times N})\) the sector with fermion number
\(R\), the Hilbert space decomposes as
\[
\mathcal H=\bigoplus_{R=0}^{N^2}V_R,
\qquad
\dim V_R=\binom{N^2}{R}.
\]
By construction, \(F|_{V_R}=R\).
Since \(Q_p\) raises fermion number by \(p\),
its sectorwise action is a map
\begin{equation}\label{eq:QR}
Q_R:V_R\to V_{R+p}.
\end{equation}
We use the convention \(V_R=0\) outside \(0\le R\le N^2\), so
\(Q_R=0\) whenever its source or target sector lies outside this range.
The Hamiltonian restricted to \(V_R\) is
\begin{equation}\label{eq:HR}
H_R=Q_R^{\dagger}Q_R+Q_{R-p}Q_{R-p}^{\dagger}.
\end{equation}
Particle--hole conjugation is the unitary Hodge map on the exterior
algebra, sending each occupation pattern to its complement.  Its phase
is chosen to exchange creation and annihilation operators.  It sends
\(V_R\) to \(V_{N^2-R}\) and
conjugates \(Q_p\) to \(Q_p^\dagger\) up to an irrelevant phase.
Consequently it maps \(H_R\) to \(H_{N^2-R}\), giving
\begin{equation}\label{eq:phsym}
h_R=h_{N^2-R}
\end{equation}
for the multiplicities \(h_R\) at zero energy.
The BPS generating function graded by fermion number is
\begin{equation}\label{eq:Zbps-def}
\Zbps(x):=\sum_{R=0}^{N^2} h_R\,x^R,
\end{equation}
where \(x\) is the fugacity for fermion number.
We write
\begin{equation}\label{eq:Zbps-total}
\Zbps:=\Zbps(1)=\sum_{R=0}^{N^2}h_R
\end{equation}
for the total BPS count, and use the unevaluated symbol \(\Zbps(x)\)
whenever the charge grading matters.
The minimum BPS charge \(q_{\min}\) is the smallest~\(R\) with \(h_R>0\).

\subsection{Quadratic Casimir and the spectrum}
\label{sec:casimir}

For \(p=3\), the supercharge \(Q_3\) on the traceless sector
\(\Lambda^\bullet W_0\simeq
\Lambda^\bullet(\mathfrak{su}(N)\otimes_{\mathbb R}\mathbb C)\)
is wedge multiplication by the cubic primitive cocycle.  Here a primitive cocycle is an
indecomposable generator of the cohomology of the Lie algebra.  In this
model, it is the invariant antisymmetric tensor whose components are
the coefficients of \(\tr(A^3)\).  In the normalization used here,
Chen's result is
\begin{equation}\label{eq:cubicH}
H=3N(N^2-1)-9\hat C_2.
\end{equation}
Here \(\hat C_2\) is the quadratic Casimir for the conjugation action on
the full Fock space.  The trace mode is an \(\SU(N)\) singlet, and
\(\hat C_2\) is normalized to have eigenvalue \(N\) on the adjoint
representation.
Kostant studied the corresponding Casimir and Laplacian spectrum on the
exterior algebra of a compact Lie algebra~\cite{Kostant1965}.
Every state in a fixed irreducible summand therefore has the same
energy.  Representation theory can consequently determine the cubic
spectrum, as in Ref.~\cite{Chen2025}.

For odd \(p\ge5\) with \(p\le2N-1\), the trace fermion
still decouples exactly (Section~\ref{sec:U1}), and the traceless
supercharge is defined by the primitive cocycle of degree \(p\)
\cite{ChevalleyEilenberg,Koszul1950}.  The associated Laplacian is not
in general a function only of the central Casimir operators
\cite{ChryssomalakosEtAl1999,deAzcarragaMacfarlane2000}.  The numerical
example in Section~\ref{sec:casimir-example} makes this failure
explicit.  Two equivalent irreducible representations occur at
different energies.

\paragraph{Cartan generators.}
Adding the Cartan generators does not resolve repeated irreducible summands.
They commute with \(H\), since \(H\) is invariant under \(\SU(N)\), and
label the weight of a state inside an irreducible representation.
A repeated representation, however, appears as
\(\C^{m}\otimes V_\lambda\) with \(m\ge2\), and every element of the
Cartan subalgebra acts as
\(\mathbf 1_{\C^m}\otimes\rho_\lambda(X)\) for
\(X\) in the Cartan subalgebra.
It acts trivially on the multiplicity space \(\C^m\).
The Cartan generators therefore do not distinguish the copies.
Splitting such copies requires an operator acting nontrivially on
\(\C^m\).  For \(p=5\), Section~\ref{sec:normalH} shows that only
\(K_3\) and \(K_4\) can carry spectral information not fixed by
\(\hat n\) and \(\hat C_2\).  In the \((5,4)\), \(R=4\) example,
\(\hat n\) and \(\hat C_2\) have the same values on the two equivalent
irreducible components, while the combination \(-K_3+K_4\) in
\(\widetilde H\) produces the observed splitting.
Whether the model admits a family of \(\SU(N)\)-invariant
conserved operators, beyond functions of \(H\), that is defined for
all \(N\ge3\) and resolves these multiplicities remains open.

\subsection{Cohomological structure}
\label{sec:cohom}

We first prove the trace identities needed for the integral
normalization of \(Q_p\) and for the exact decoupling of the trace
fermion in Section~\ref{sec:U1}.

\begin{proposition}[Traces of powers of an anticommuting matrix]
\label{prop:cyclic}
Let \(X\) be an \(N\times N\) matrix whose entries are odd generators of an exterior algebra.
Then
\begin{equation}\label{eq:cyclic}
\tr(X^p)=
\begin{cases}
0,& p \text{ even},\\[2mm]
p\,\widetilde T_p(X),& p \text{ odd},
\end{cases}
\end{equation}
where \(\widetilde T_p(X)\) is integral in the ordered monomial basis of the exterior algebra.
\end{proposition}

\begin{proof}
A nonzero term in \(\tr(X^p)=\sum_{i_0,\ldots,i_{p-1}}X_{i_0 i_1}X_{i_1 i_2}\cdots X_{i_{p-1}i_0}\)
corresponds to an ordered tuple of \(p\) distinct matrix entries
\(e_k=(i_k,i_{k+1})\), with the indices read modulo~\(p\).  Any repeated
entry gives a vanishing wedge product.
The cyclic rotation \(C_p:(e_0,e_1,\ldots,e_{p-1})\mapsto (e_1,\ldots,e_{p-1},e_0)\) acts on the set of such tuples.
If a nontrivial rotation \(C_p^d\) (\(0<d<p\)) fixed a tuple, then \(e_0=e_d\), contradicting distinctness.
Thus every nonzero orbit has size exactly~\(p\).
Within an orbit, a single rotation permutes \(p\) odd generators cyclically, contributing a sign \((-1)^{p-1}\).
For even~\(p\), this gives alternating signs whose sum is zero.  For
odd~\(p\), all \(p\) terms in each orbit have the same sign, so the
orbit sum is~\(p\) times a single ordered monomial with integer
coefficient.
\end{proof}

\begin{remark}
The even case recovers the classical identity
\(\tr(X^{2k})=0\) for matrices with anticommuting entries
\cite{Itoh2015}.  The same cyclic orbit argument proves both statements.
No use of invariant theory is required.
\end{remark}

The vanishing of the even traces uses only that the matrix entries are
odd elements of an exterior algebra.  It does not require independent
generators and therefore applies to the traceless matrix \(A\) as well.

For odd~\(p\), Proposition~\ref{prop:cyclic} applied to \(\Psi\) gives the integral reduction
\begin{equation}\label{eq:Qreduced}
Q_p=\tr(\Psi^p)=p\,\widetilde Q_p .
\end{equation}
Multiplication by \(\widetilde Q_p\) has sector matrices
\(\widetilde Q_{p,R}\) with integer entries in the ordered bases of
matrix unit monomials.
Since \(\widetilde Q_p\) is nilpotent (see below), these maps define a cochain complex over~\(\Z\).

The next structural fact is that the BPS problem is cohomological for every odd \(p\), not just for \(p=3\).

\begin{proposition}\label{prop:nilpotent}
For every odd \(p\), the supercharge \(Q_p=\tr(\Psi^p)\) satisfies \(Q_p^2=0\).
\end{proposition}

\begin{proof}
The operator \(Q_p\) is a sum of monomials, each of which is a product of exactly \(p\) of the mutually anticommuting creation operators \(\Psi_{ij}\).
The square \(Q_p^2\) is therefore a sum over pairs of such monomials
\(m_1,m_2\).  Terms with \(m_1=m_2\) vanish because the square of an odd
monomial is zero.  For \(m_1\ne m_2\), both monomials have odd degree,
so
\[
m_1m_2=(-1)^{p^2}m_2m_1=-m_2m_1.
\]
The terms associated with \((m_1,m_2)\) and \((m_2,m_1)\) cancel.  Thus
\(Q_p^2=0\) over every field of characteristic different from two.
The argument uses only the odd degree of \(Q_p\), not its trace form.
The same cancellation is an integral identity.  Since
\(Q_p^2=p^2\widetilde Q_p^2=0\) and the integral exterior algebra has
no torsion, \(\widetilde Q_p^2=0\) over~\(\Z\).
\end{proof}

The integral normalization gives an overall factor \(p^2\) in the
Hamiltonian:
\begin{equation}\label{eq:Hreduced}
H=\{Q_p,Q_p^\dagger\}=p^2\,\widetilde H, \qquad \widetilde H:=\{\widetilde Q_p,\widetilde Q_p^\dagger\}.
\end{equation}
The numerical energies below are consistent with this normalization.

Since \(Q_p^2=0\), the operator \(Q_p\) is a differential.
Within each residue class \(R\bmod p\), the sectors
\(V_R\xrightarrow{Q_R}V_{R+p}\xrightarrow{Q_{R+p}}\cdots\)
form a cochain complex.
The BPS multiplicity in sector \(R\) is the dimension of the cohomology in degree~\(R\),
\begin{equation}\label{eq:cohom}
h_R=\dim V_R-\operatorname{rank}Q_R-\operatorname{rank}Q_{R-p}.
\end{equation}
We write
\begin{equation}\label{eq:rank-def}
r_R:=\operatorname{rank}Q_R .
\end{equation}
Thus the multiplicities at zero energy can be obtained from the ranks of the rectangular maps \eqref{eq:QR}, without diagonalizing the full Hamiltonian.

Wedging elements of complementary degree gives a pairing in the highest
exterior degree and the following symmetry of the sector ranks.  Here
\([c]_{\Lambda^{N^2}}\) denotes the coefficient of
the ordered top form in \(c\).

\begin{theorem}[Symmetry of the sector ranks]
\label{thm:rankpal}
For every odd~\(p\) with \(p\le2N-1\), the pairing in the highest
exterior degree \(\langle a,b\rangle=[a\wedge b]_{\Lambda^{N^2}}\) satisfies
\[
\langle Q_R\,a,\,b\rangle \;=\; (-1)^{pR}\,\langle a,\,Q_{N^2-p-R}\,b\rangle
\]
for all \(a\in V_R\), \(b\in V_{N^2-p-R}\).
Since the pairing is nondegenerate, \(Q_R\) and \(Q_{N^2-p-R}\) have equal rank:
\begin{equation}\label{eq:rank-palindrome}
r_R = r_{N^2-p-R}.
\end{equation}
Consequently the polynomial
\(\mathcal R_{p,N}(x)=\sum_R r_R\,x^R\) is palindromic of degree
\(N^2-p\).
\end{theorem}

\begin{proof}
Writing \(q\) for the element of degree \(p\) that acts by wedge
multiplication, we have
\begin{align*}
\langle Q_Ra,b\rangle
&=[q\wedge a\wedge b]_{\Lambda^{N^2}}\\
&=(-1)^{pR}[a\wedge q\wedge b]_{\Lambda^{N^2}}\\
&=(-1)^{pR}\langle a,Q_{N^2-p-R}b\rangle .
\end{align*}
The second line moves \(q\) past the element \(a\), which has degree \(R\).
The pairing \(V_R\otimes V_{N^2-R}\to\C\) is nondegenerate, so the adjoint relation forces \(\operatorname{rank}Q_R=\operatorname{rank}Q_{N^2-p-R}\).
For \(p\le2N-1\), the description of the invariant generators in
Appendix~\ref{app:cohom-background} gives \(Q_p\ne0\).  Since
\(\dim V_0=1\), the restriction \(Q_0\) has rank one.
Equation~\eqref{eq:rank-palindrome} then gives \(r_{N^2-p}=1\).  For
\(R>N^2-p\), the target of \(Q_R\) is zero.  Hence
\(\mathcal R_{p,N}(x)\) has degree exactly \(N^2-p\).
\end{proof}

Appendix~\ref{app:cohom-background} compares this complex with the
Chevalley--Eilenberg and Aomoto differentials
\cite{ChevalleyEilenberg,Aomoto1975,OrlikTerao2001}.  It also explains
why the present differential is not the usual Lie algebra differential.

\subsection{\texorpdfstring{$U(1)$}{U(1)} decoupling and the condition \texorpdfstring{$p\le2N-1$}{p <= 2N-1}}
\label{sec:U1}

\begin{theorem}[Exact $U(1)$ decoupling]
\label{thm:U1}
With the orthogonal decomposition \eqref{eq:trace-decomposition}, the \(U(1)\) trace fermion \(\chi\) decouples exactly for every odd \(p\ge 3\):
\begin{equation}\label{eq:trace-decouple}
Q_p=\tr(\Psi^p)=\tr(A^p).
\end{equation}
Consequently the cochain complex splits as
\(\Lambda^\bullet(\C T_0)\otimes
\bigl(\Lambda^\bullet(W_0),\,Q_p\bigr)\),
and the BPS generating function acquires one exact factor of \((1+x)\) from the free \(\chi\) mode.
Any further factor of \(1+x\) comes from the traceless complex.
\end{theorem}

\begin{proof}
Set $X:=(\chi/\sqrt{N})\mathbf{1}$.
Since $\chi$ is an odd Grassmann generator, $X^2=0$.
Since $\chi$ and each $A_{ij}$ are both odd, they anticommute: $XA=-AX$.
Using $AX=-XA$ repeatedly gives $A^jX=(-1)^jXA^j$,
so the terms containing one factor of $X$ in the expansion of $(A+X)^p$ sum to
\[
  \sum_{j=0}^{p-1}A^jXA^{p-1-j}
  =\Bigl(\sum_{j=0}^{p-1}(-1)^j\Bigr)XA^{p-1}=XA^{p-1},
\]
where the last equality uses $\sum_{j=0}^{p-1}(-1)^j=1$ for odd~$p$.
Terms containing at least two factors of $X$ vanish after the factors
are anticommuted next to one another, since $X^2=0$.  Hence
$(A+X)^p=A^p+XA^{p-1}$.
Taking the trace: $\tr(\Psi^p)=\tr(A^p)+\frac{\chi}{\sqrt{N}}\tr(A^{p-1})$.
Since $p-1$ is even, $\tr(A^{p-1})=0$ by Proposition~\ref{prop:cyclic},
giving~\eqref{eq:trace-decouple}.
The complex then splits as a tensor product with the free mode~$\chi$.
By the K\"{u}nneth formula~\cite{Weibel1994}, each traceless cohomology
class occurs once with this mode empty and once with it occupied,
giving the factor $(1+x)$.
\end{proof}

\begin{corollary}[Factors implied by the complex]
\label{cor:structural-factors}
For every odd \(p\ge3\) and every \(N\ge2\),
\begin{equation}\label{eq:structural-factors}
(1+x)^2\mid \Zbps^{(p,N)}(x).
\end{equation}
If \(N\) is odd, then
\begin{equation}\label{eq:structural-factors-odd}
(1+x)^3\mid \Zbps^{(p,N)}(x).
\end{equation}
The corresponding polynomial for the map ranks,
\[
\mathcal R_{p,N}(x)=\sum_R\operatorname{rank}Q_R\,x^R
\]
satisfies
\begin{equation}\label{eq:trace-rank-splitting}
\mathcal R_{p,N}(x)=(1+x)\mathcal R^0_{p,N}(x),
\end{equation}
where \(\mathcal R^0_{p,N}\) is the corresponding polynomial for the
traceless complex.  If \(N\) is odd and \(p\le2N-1\), then
\((1+x)^2\mid\mathcal R_{p,N}(x)\).
\end{corollary}

\begin{proof}
Let
\[
Z_0(x)=\sum_R \dim H^R\!\left(\Lambda^\bullet W_0,Q_p\right)x^R
\]
be the generating polynomial for the traceless complex.  Theorem~\ref{thm:U1}
gives \(\Zbps(x)=(1+x)Z_0\).  Since \(Q_p\) has odd degree, the value
\(Z_0(-1)\) is the Euler characteristic of the traceless complex.  Hence
\[
Z_0(-1)=\sum_R(-1)^R\binom{N^2-1}{R}=(1-1)^{N^2-1}=0.
\]
This proves~\eqref{eq:structural-factors}.

The pairing in the highest exterior degree, applied to
\(\Lambda^\bullet W_0\), gives
\(Z_0(x)=x^{N^2-1}Z_0(x^{-1})\).  If \(N\) is odd, the exponent
\(N^2-1\) is even.  For a polynomial with this symmetry and an even
exponent, a zero at \(x=-1\) has even multiplicity.
Thus \((1+x)^2\mid Z_0\), which proves
\eqref{eq:structural-factors-odd}.

In degree \(R\), the trace splitting gives
\[
\Lambda^R(W_0\oplus\C T_0)
=\Lambda^R W_0\oplus
\bigl(\chi\wedge\Lambda^{R-1}W_0\bigr).
\]
The differential has the two blocks \(Q_R^0\) and \(-Q_{R-1}^0\).
Writing \(s_R=\operatorname{rank}Q_R^0\), we obtain
\(\operatorname{rank}Q_R=s_R+s_{R-1}\), which proves
\eqref{eq:trace-rank-splitting}.  When \(N\) is odd, rank symmetry makes
\(\mathcal R_{p,N}\) reciprocal of even degree \(N^2-p\).  Its zero at
\(-1\) therefore has even multiplicity, which supplies the second
rank factor.
\end{proof}

The supercharge is nonzero only when \(p\le2N-1\).  For \(p>2N-1\),
the trace \(\tr(A^p)\) vanishes identically
\cite{Itoh2015}.  Thus \(p=3,5,7\) first occur at \(N=2,3,4\),
respectively.  In particular, \(Q_7=0\) at \(N=3\), so every state is
BPS.  Appendix~\ref{app:cohom-background} explains the same condition in terms of
primitive cocycles \cite{ChevalleyEilenberg,Koszul1950}.

When the cochain description is used, each residue class \(a\in \Z/p\Z\) has Euler characteristic
\begin{equation}\label{eq:index}
I_a(N^2)=\sum_{R\equiv a\, (\mathrm{mod}\,p)}(-1)^R\binom{N^2}{R}
=\frac{1}{p}\sum_{\ell=0}^{p-1}\omega^{-a\ell}(1-\omega^{\ell})^{N^2},
\qquad \omega=e^{2\pi i/p}.
\end{equation}
These are the Witten indices of the \(p\) complexes obtained by fixing
fermion number modulo \(p\).  They count BPS states with signs and can
be computed from binomial coefficients without knowing the spectrum.
Since \(p\) is odd, \((-1)^{R+p}=-(-1)^R\).  On summing
\eqref{eq:cohom} over \(R\equiv a\pmod p\) with weights \((-1)^R\),
each \(\operatorname{rank}Q_R\) occurs twice with opposite signs and
cancels.  Hence
\(I_a(N^2)=\sum_{R\equiv a\, (\mathrm{mod}\,p)}(-1)^R h_R\).
Their absolute values give an exact lower bound on the total BPS count,
\begin{equation}\label{eq:indexfloor}
\sum_R h_R\ge \sum_{a=0}^{p-1}|I_a|.
\end{equation}
At fixed~\(N\), we record the difference between the computed count and
this bound.
This notion is separate from fortuity under changes of rank.

Equation~\eqref{eq:index} determines this bound whenever \(p\le2N-1\).
For fixed \(p\), at large~\(N\) it grows exponentially in~\(N^2\) with rate \(\log\bigl(2\cos(\pi/(2p))\bigr)\).
It therefore gives a rigorous lower bound on \(Z_{\rm BPS}\) with the same rate.
The derivation is given in Section~\ref{sec:largeN}.

\section{Hamiltonian after normal ordering}
\label{sec:normalH}

The cohomological description determines the multiplicities at zero
energy without displaying the interactions in the Hamiltonian.  We now
expand the reduced Hamiltonian after normal ordering.  We first derive a
formula valid for every odd \(p\).  We then obtain \(K_0,K_1,K_2\) for
\(p=5\) and every \(N\ge3\), and compare the commuting terms at \(N=3\)
with the nonzero commutators found at \(N=4\).

\subsection{General formula for odd \texorpdfstring{$p$}{p}}

The reduced supercharge creates from the vacuum an invariant state of
\(p\) fermions under \(\SU(N)\),
\[
|\Omega_p\rangle:=\widetilde Q_p|0\rangle .
\]
For any state \(|\Phi\rangle\) in the Fock space,
\begin{equation}\label{eq:add-remove}
\langle\Phi|\widetilde H|\Phi\rangle
=\|\widetilde Q_p|\Phi\rangle\|^2
+\|\widetilde Q_p^\dagger|\Phi\rangle\|^2 .
\end{equation}
Thus the Hamiltonian measures whether this invariant state of \(p\)
fermions can be added to or removed from \(|\Phi\rangle\).  A BPS
state allows neither process.

We now separate these processes by the number of creation and
annihilation operators.  A term labelled by \(k\) contains \(k\) of
each.
As in supersymmetric SYK models whose supercharge is built from monomials
of degree \(p\) in the fermions~\cite{FuGaiotto}, the largest possible interaction contains
\(2p-2\) fermion operators.
Here the couplings are fixed by the trace and by \(\SU(N)\) symmetry
rather than chosen randomly.
For standard conventions in second quantization and normal ordering, see
Ref.~\cite{FetterWalecka1971}.

Let \(W=\mathrm{Mat}_N(\C)\) with the Hilbert--Schmidt inner product.
Use the orthonormal Hermitian basis
\(\{T_0,T_1,\ldots,T_{N^2-1}\}\) fixed in
Section~\ref{sec:fock}, with
\(T_0=\mathbf 1/\sqrt N\) and \(T_\alpha\in W_0\), and write
\[
\Psi=\sum_{\mu=0}^{N^2-1}a_\mu^\dagger T_\mu,
\qquad a_0^\dagger=\chi,\qquad a_\alpha^\dagger=b_\alpha^\dagger,
\qquad a_\mu=(a_\mu^\dagger)^\dagger,
\qquad \{a_\mu,a_\nu^\dagger\}=\delta_{\mu\nu}.
\]
The invariant complex bilinear trace pairing
\(B(X,Y)=\tr(XY)\) also satisfies \(B(T_\mu,T_\nu)=\delta_{\mu\nu}\).
It therefore identifies \(W\) with \(W^*\) equivariantly under
conjugation by \(\SU(N)\).  We use this identification below.
If \(I=\{i_1<\cdots<i_k\}\), define
\[
a_I^\dagger:=a_{i_1}^\dagger\cdots a_{i_k}^\dagger,
\qquad
a_I:=(a_I^\dagger)^\dagger=a_{i_k}\cdots a_{i_1}.
\]
By Theorem~\ref{thm:U1}, the reduced supercharge of \eqref{eq:Qreduced} is
\begin{equation}\label{eq:omega-expansion}
\widetilde Q_p=\frac{1}{p}\tr(A^p)
=\sum_{|S|=p}(\Omega_p)_S\,a_S^\dagger .
\end{equation}
In the chosen Hermitian basis, complex conjugation reverses the order
inside each trace.  Therefore
\[
\overline{(\Omega_p)_S}
=(-1)^{p(p-1)/2}(\Omega_p)_S .
\]
Thus the components are real for \(p\equiv1\pmod4\) and purely imaginary
for \(p\equiv3\pmod4\).
Under the identification defined by the trace pairing, the coefficients
\((\Omega_p)_S\) are the components of an \(\SU(N)\)-invariant form
\(\Omega_p\in\Lambda^pW^*\).  Exact decoupling of the trace mode implies that
this form is supported on \(W_0\).

Equip the exterior powers of \(W\) and \(W^*\) with their induced
Hermitian inner products.  The following map records the
\((p-k)\)-particle wavefunction left after removing \(k\) particles from
\(|\Omega_p\rangle\):
\begin{equation}\label{eq:contraction-map}
\mathcal A_k:\Lambda^kW\longrightarrow\Lambda^{p-k}W^*,\qquad
\mathcal A_k(\alpha)=\iota_\alpha\Omega_p ,
\qquad 0\le k\le p .
\end{equation}
Here \(\iota_\alpha\Omega_p\) denotes contraction by the multivector
\(\alpha\).  Thus, if \(\alpha=v_1\wedge\cdots\wedge v_k\), then
\((\iota_\alpha\Omega_p)(w_1,\ldots,w_{p-k})
=\Omega_p(v_1,\ldots,v_k,w_1,\ldots,w_{p-k})\), extended linearly.
For disjoint index subsets \(I,T\), each written in increasing order,
the matrix elements of \(\mathcal A_k\) are
\[
(\mathcal A_k)_{T,I}
=\epsilon(T,I)(\Omega_p)_{T\cup I},
\qquad |I|=k,\quad |T|=p-k,
\]
where \(\epsilon(T,I)\) is the sign that orders the concatenated wedge product.
The form \(\Omega_p\) vanishes whenever one of its arguments lies in the trace direction.
The common phase of its components implies that
\(\mathcal A_k^\dagger\mathcal A_k\) is real symmetric.

The contraction matrices have a direct interpretation in terms of
correlations.  With the ordering of subsets fixed above, define
\begin{equation}\label{eq:rdm-kernel}
\Gamma^{(k)}_{I,J}
:=\langle\Omega_p|a_I^\dagger a_J|\Omega_p\rangle
=(\mathcal A_k^\dagger\mathcal A_k)_{I,J}.
\end{equation}
Thus \(\Gamma^{(k)}\) is the unnormalized reduced density matrix for
\(k\) particles in \(|\Omega_p\rangle\), in the index convention of
\eqref{eq:rdm-kernel}~\cite{Coleman1963}, and
\begin{equation}\label{eq:rdm-trace}
\tr\Gamma^{(k)}
=\binom{p}{k}\|\Omega_p\|^2 .
\end{equation}
The operator \(K_k\) applies this correlation matrix to the full Fock
space:
\begin{equation}\label{eq:Kk-definition}
K_k:=
\sum_{\substack{|I|=k\\|J|=k}}
(\mathcal A_k^\dagger\mathcal A_k)_{I,J}\,
a_I^\dagger a_J .
\end{equation}
The strings \(a_I^\dagger a_J\), with \(|I|=|J|\), form a basis of
operators after normal ordering that preserve fermion number.  The
expansion is therefore complete.  In particular, \(K_k\) is the
contribution with \(k\) creation and \(k\) annihilation operators.

For a single monomial the alternating structure has a simple
interpretation in terms of Pauli blocking.
If \(q=a_1^\dagger\cdots a_p^\dagger\) and
\(n_i=a_i^\dagger a_i\), then
\begin{equation}\label{eq:monomial-blocking}
\{q,q^\dagger\}
=\prod_{i=1}^p n_i+\prod_{i=1}^p(1-n_i).
\end{equation}
The two terms project onto the configurations in which all \(p\) modes
are occupied or all are empty.
Expanding the second product gives an alternating sum in the occupied
modes.  For odd \(p\), its term of degree \(p\) cancels the first product.
The following proposition is the corresponding statement for the
supercharge \(Q_p\) defined by one trace.

\begin{proposition}[Hamiltonian after normal ordering]\label{prop:normalH}
For every odd \(p\),
\begin{equation}\label{eq:normalH}
\widetilde H=\{\widetilde Q_p,\widetilde Q_p^\dagger\}
=\sum_{k=0}^{p-1}(-1)^k K_k .
\end{equation}
Each \(K_k\) is positive semidefinite, preserves fermion number, and is separately \(\SU(N)\)-invariant.
Equivalently,
\(H=p^2\sum_{k=0}^{p-1}(-1)^kK_k\).
\end{proposition}

\begin{proof}
For \(0\le k\le p-1\), only the normal ordering of
\(\widetilde Q_p^\dagger\widetilde Q_p\) contributes a term with
\(k\) creation and \(k\) annihilation operators.  Choosing the
\(p-k\) contracted modes and moving the \(k\) remaining creation operators
past the \(k\) remaining annihilation operators gives
\((-1)^{k^2}=(-1)^k\).  Hence, for \(|I|=|J|=k\),
\[
\begin{aligned}
\operatorname{coeff}_{a_I^\dagger a_J}\widetilde H
&=(-1)^{k^2}
\sum_{\substack{|T|=p-k\\T\cap(I\cup J)=\varnothing}}
\overline{\epsilon(T,J)(\Omega_p)_{T\cup J}}\,
\epsilon(T,I)(\Omega_p)_{T\cup I}\\
&=(-1)^k(\mathcal A_k^\dagger\mathcal A_k)_{J,I}
=(-1)^k(\mathcal A_k^\dagger\mathcal A_k)_{I,J},
\qquad 0\le k\le p-1 .
\end{aligned}
\]
The last equality uses the real symmetry proved above.  For \(k=p\),
the term from \(\widetilde Q_p^\dagger\widetilde Q_p\) is
\((-1)^{p^2}\) times the term from
\(\widetilde Q_p\widetilde Q_p^\dagger\), so the two cancel because
\(p\) is odd.
This proves \eqref{eq:normalH}.
The Gram form makes each \(K_k\) positive semidefinite, and invariance follows from that of \(\Omega_p\).
\end{proof}

\subsection{The quintic Hamiltonian}

In the quintic model, \(K_0\) is the constant term.  The operators
\(K_1\) and \(K_2\) contain respectively one and two creation operators,
together with the same number of annihilation operators, in the
expansion obtained by normal ordering \(\widetilde H\).
All three are fixed analytically by \(\SU(N)\) representation theory,
normalization, and the contraction identities for the primitive form.
Let \(\hat C_2\) be the quadratic Casimir for the conjugation action of
\(\SU(N)\) on the full Fock space.  The trace mode is an \(\SU(N)\)
singlet, and \(\hat C_2\) is normalized to have eigenvalue \(N\) on the
adjoint representation.
For \(p=5\) and every \(N\ge 3\), the following identities hold as
operators on the full Fock space:
\begin{align}
K_0&=\frac{N(N^2-1)(N^2-4)}{5},\label{eq:K0p5}\\
K_1&=N(N^2-4)\,\hat n,\label{eq:K1p5}\\
K_2&=N\hat n(\hat n-1)
 +(N^2-6)(N\hat n-\hat C_2).\label{eq:K2p5}
\end{align}
All adjoints and norms below use the Fock inner product induced by the
Hilbert--Schmidt inner product on \(W\).
The scale of the invariant form is fixed by
\[
\widetilde Q_5=\frac{1}{5}\tr(A^5)
=\frac{1}{5}\tr(\Psi^5).
\]
Its coefficients in the original wedge basis built from matrix units are
integers by \eqref{eq:Qreduced}.
For comparison with Ref.~\cite{deAzcarragaMacfarlaneCocycles}, set
\(\lambda_a=\sqrt2\,T_a\), so that
\(\tr(\lambda_a\lambda_b)=2\delta_{ab}\), as in that reference.
Let \(f_{abc}\) be the real structure constants defined by
\([\lambda_a,\lambda_b]=2i f_{abc}\lambda_c\), equivalently
\([T_a,T_b]=i\sqrt2 f_{abc}T_c\).
We use normalized antisymmetrization,
\[
T_{[a_1\cdots a_5]}:=\frac{1}{5!}
\sum_{\sigma\in S_5}\operatorname{sgn}(\sigma)
T_{a_{\sigma(1)}}\cdots T_{a_{\sigma(5)}}.
\]
Let \(\Omega^{\rm lit}\) denote the invariant cocycle of degree five in
the normalization used in Ref.~\cite{deAzcarragaMacfarlaneCocycles}.
With this convention, equation~(117) of that reference gives
\(\tr\lambda_{[a_1\cdots a_5]}=-2\Omega^{\rm lit}_{a_1\cdots a_5}\).
Collecting terms in
\(\tr(A^5)/5\) into increasing index sets therefore gives
\begin{equation}\label{eq:omega-normalization}
(\Omega_5)_{a_1\cdots a_5}
=\frac{5!}{5}\tr T_{[a_1\cdots a_5]}
=-6\sqrt2\,\Omega^{\rm lit}_{a_1\cdots a_5},
\qquad a_1<\cdots<a_5 .
\end{equation}
After this normalization, equations~(73), (76), (82), and (84) of that
reference give
\begin{align}
\mathcal A_1^\dagger\mathcal A_1
&=N(N^2-4)\,\mathbf 1_{\rm adj},\label{eq:A1contraction}\\
\mathcal A_2^\dagger\mathcal A_2
&=N(N^2-4)P_{\rm adj}+2N P_{\perp}.\label{eq:A2contraction}
\end{align}
The first identity holds on \(W_0\), and the second on
\(\Lambda^2W_0\).  Here \(P_{\rm adj}\) projects onto the adjoint summand
and \(P_\perp\) onto its orthogonal complement
~\cite{ChryssomalakosEtAl1999,deAzcarragaMacfarlaneCocycles}.
In translating equation~(84), we used
\((P_{\rm adj})_{ab,cd}=2N^{-1}f_{abe}f_{cde}\) in the wedge basis
indexed by increasing pairs \(a<b\) and \(c<d\).
The first identity gives \eqref{eq:K1p5}.
Since \(K_0=\lVert\Omega_5\rVert^2\) and
\(\tr(\mathcal A_1^\dagger\mathcal A_1)=5\lVert\Omega_5\rVert^2\),
taking the trace gives \eqref{eq:K0p5}.
For \(N\ge3\), the orthogonal complement of the adjoint representation
in \(\Lambda^2W_0\) is a conjugate pair of irreducible representations.
Indeed, writing \(W=V\otimes V^*\) with \(V=\C^N\), one has
\[
\Lambda^2(V\otimes V^*)
\simeq
(S^2V\otimes\Lambda^2V^*)
\oplus(\Lambda^2V\otimes S^2V^*)
\simeq 2\,{\rm adj}\oplus V_\lambda\oplus\overline{V_\lambda},
\]
where \(V_\lambda\) is the irreducible representation with highest
weight \(\lambda=(2,0,\ldots,0,-1,-1)\), written in the standard
\(N\)-tuple notation, and \(\overline{V_\lambda}\) is its conjugate.
Since \(V\otimes V^*\simeq\C\oplus W_0\),
it follows that
\[
\Lambda^2W_0\simeq
{\rm adj}\oplus V_\lambda\oplus\overline{V_\lambda}.
\]
For a highest weight \(\ell=(\ell_1,\ldots,\ell_N)\), with
\(\sum_i\ell_i=0\), the quadratic Casimir in the present normalization
is
\[
C_2(\ell)=\frac12\sum_{i=1}^N
\ell_i\bigl(\ell_i+N+1-2i\bigr).
\]
This gives \(N\) for the adjoint weight
\((1,0,\ldots,0,-1)\) and \(2N\) for \(\lambda\) and its conjugate.
Since \(W=W_0\oplus\C T_0\),
\[
\Lambda^2W=\Lambda^2W_0\oplus(T_0\wedge W_0).
\]
On \(\Lambda^2W_0\), \(\hat n=2\), and the right side of
\eqref{eq:K2p5} has eigenvalues \(N(N^2-4)\) and \(2N\) on the
adjoint and its complement.
On \(T_0\wedge W_0\), \(\mathcal A_2^\dagger\mathcal A_2=0\).
There \(\hat n=1\) and \(\hat C_2=N\), so the right side also
vanishes.
Each generator of the conjugation action contains one creation and one
annihilation operator.  Since \(\hat C_2\) is quadratic in these
generators, its expansion after normal ordering contains no term with
more than two creation and two annihilation operators.
After normal ordering, \(\hat n(\hat n-1)\) contains only terms with two
creation and two annihilation operators.  On the traceless sector with
one fermion, \(\hat C_2=N\mathbf 1\), while it vanishes on the trace mode.
Its part with one creation and one annihilation operator is therefore
\(N\hat n\), so \(N\hat n-\hat C_2\) also contains only terms with two
creation and two annihilation operators.  For any such operator, its
matrix elements in the ordered basis \(a_I^\dagger|0\rangle\) of
\(\Lambda^2W\) are exactly the coefficients multiplying
\(a_I^\dagger a_J\).  Its action on \(\Lambda^2W\)
therefore determines the operator on the full Fock space.  Since the two
sides of \eqref{eq:K2p5} agree on both summands of \(\Lambda^2W\), the
identity follows.

Thus the contribution for one particle treats every traceless mode equally.
For a pair, the strength is \(N(N^2-4)\) in the adjoint channel and
\(2N\) in its complement.
These are exact operator identities for every \(N\ge3\).

It follows that
\begin{equation}\label{eq:Hlowp5}
\widetilde H_{\le2}:=K_0-K_1+K_2
=\frac{N(N^2-1)(N^2-4)}{5}
+N\hat n(\hat n-3)-(N^2-6)\hat C_2 .
\end{equation}
Because \(\Omega_p\) vanishes on the trace direction,
\(\mathcal A_k^\dagger\mathcal A_k\) is supported on
\(\Lambda^kW_0\).  Hence \(K_k\) commutes with \(n_\chi\).
Together with conservation of the total fermion number, this gives
\([K_k,\hat n]=0\).
Its \(\SU(N)\)-invariance also gives \([K_k,\hat C_2]=0\).
For \(p=5\) and every \(N\ge3\), the identities
\eqref{eq:K0p5}--\eqref{eq:K2p5} and the commutation relations above
show that each of \(K_0\), \(K_1\) and \(K_2\) commutes with every
\(K_j\).  Their combination \(\widetilde H_{\le2}\) therefore also
commutes with every \(K_j\).  Consequently, the only commutator among
the five terms that can be nonzero is \([K_3,K_4]\).
For \(p=5\),
\[
\widetilde H=\widetilde H_{\le2}-K_3+K_4.
\]
Thus only \(K_3\) and \(K_4\) can contain spectral information not fixed
by the traceless fermion number and the quadratic Casimir.
At fixed \(\hat n\), \(\widetilde H_{\le2}\) is scalar on each
\(\SU(N)\) irreducible representation.
An invariant term with more creation and annihilation operators can
nevertheless act on the multiplicity
space when the same representation occurs more than once.
This is how \(K_3\) and \(K_4\) can split states
with identical \((\hat n,\hat C_2)\).

For comparison, consider the cubic model.  Write \(K_k^{(3)}\) for the
corresponding cubic operators.
For \(a<b<c\), the cubic form has components
\[
(\Omega_3)_{abc}=\tr\!\bigl(T_a[T_b,T_c]\bigr)
=i\sqrt2\,f_{abc}.
\]
The identity \(f_{acd}f_{bcd}=N\delta_{ab}\) therefore gives
\[
\bigl(\mathcal A_1^{(3)}\bigr)^\dagger\mathcal A_1^{(3)}
=N\mathbf1_{\rm adj}.
\]
This gives \(K_1^{(3)}=N\hat n\).  Taking the trace and using
\eqref{eq:rdm-trace} gives
\(K_0^{(3)}=\|\Omega_3\|^2=N(N^2-1)/3\).
On \(\Lambda^2W_0\), the same identity for the structure constants gives
\[
\bigl(\mathcal A_2^{(3)}\bigr)^\dagger\mathcal A_2^{(3)}
=NP_{\rm adj}=2N\mathbf1-\hat C_2.
\]
On \(T_0\wedge W_0\), both
\(\bigl(\mathcal A_2^{(3)}\bigr)^\dagger\mathcal A_2^{(3)}\) and
\(N\hat n-\hat C_2\) vanish.  Both \(K_2^{(3)}\) and
\(N\hat n-\hat C_2\) contain only terms with two creation and two
annihilation operators, so their agreement on the two summands of
\(\Lambda^2W\) gives \(K_2^{(3)}=N\hat n-\hat C_2\) on the full Fock
space.
Thus
\begin{equation}\label{eq:cubic-components}
K_0^{(3)}=\frac{N(N^2-1)}{3},\qquad
K_1^{(3)}=N\hat n,\qquad
K_2^{(3)}=N\hat n-\hat C_2 .
\end{equation}
Normal ordering stops after the terms with two creation and two
annihilation operators, and the contributions from the fermion number
operator cancel:
\begin{equation}\label{eq:cubic-cancellation}
H_3=9(K_0^{(3)}-K_1^{(3)}+K_2^{(3)})
=3N(N^2-1)-9\hat C_2 .
\end{equation}
This recovers \eqref{eq:cubicH}.
The cubic Hamiltonian is determined by the quadratic Casimir because
the \(\hat n\) terms in \(K_1^{(3)}\) and \(K_2^{(3)}\) cancel, and no
\(K_k^{(3)}\) with \(k>2\) appears in \(H_3\).

Since \(K_k\) annihilates states with fewer than \(k\) traceless
fermions, the reduced quintic Hamiltonian satisfies
\(\widetilde H=\widetilde H_{\le2}\) whenever \(\hat n\le2\).
As a check, at \((5,3)\) there are no zero modes in \(V_0\) or \(V_1\),
while
\[
V_2=(\mathbf{10}\oplus\overline{\mathbf{10}})
\oplus2\cdot\mathbf{8},\qquad
H|_{V_2}=25(18-3\hat C_2).
\]
Here one \(\mathbf8\) lies in \(\Lambda^2W_0\), where \(\hat n=2\),
and the other lies in \(T_0\wedge W_0\), where \(\hat n=1\).
The displayed formula applies to both because
\(\hat n(\hat n-3)=-2\) for either value.
The relevant Casimir eigenvalues are \(\hat C_2(\mathbf 8)=3\) and
\(\hat C_2(\mathbf{10})=\hat C_2(\overline{\mathbf{10}})=6\).
This gives \(20\) BPS states and \(16\) states at \(E=225\), explaining
the onset \(h_2=20\) in \eqref{eq:h53}.

\subsection{Commutators at \texorpdfstring{$N=3$}{N=3} and
\texorpdfstring{$N=4$}{N=4}}
\label{sec:computed-commutators}

\begin{proposition}[Commutativity at \((5,3)\)]
\label{prop:su3-commutativity}
For the quintic model at \(N=3\), the five operators
\(K_0,\ldots,K_4\) commute pairwise on the full Fock space.
This follows analytically from Hodge duality between the invariant
cubic and quintic forms on \(\mathfrak{su}(3)\).
\end{proposition}

A detailed proof is given in Appendix~\ref{app:hodge-duality}.
Indeed, at \(N=3\) the traceless space has dimension eight.  Its Hodge
star maps forms of degree three to forms of degree five.  The invariant
forms obey \(\Omega_5=c\,{*}\Omega_3\) for some nonzero scalar \(c\).
Their squared norms in the present convention obey
\(\|\Omega_5\|^2=3\|\Omega_3\|^2\)
~\cite{deAzcarragaMacfarlane2000,
deAzcarragaMacfarlaneCocycles}.
Equation~\eqref{eq:su3-kernel-lift} in
Appendix~\ref{app:hodge-duality} shows that each
quintic contraction matrix is
\(3\widetilde Q_3^\dagger\widetilde Q_3\) on \(\Lambda^kW_0\).
Normal ordering the cubic anticommutator and using
\eqref{eq:cubic-cancellation} gives
\(\widetilde Q_3^\dagger\widetilde Q_3
=8-\hat C_2-\widetilde Q_3\widetilde Q_3^\dagger\).
This converts \(\widetilde Q_3^\dagger\widetilde Q_3\) into the
\(\widetilde Q_3\widetilde Q_3^\dagger\) form used in the explicit
formula there.  That formula expresses every \(K_k\)
through the commuting operators \(\hat n\), \(\hat C_2\), and
\(\widetilde Q_3\widetilde Q_3^\dagger\).
Thus the \((5,3)\) commutativity follows from this Hodge duality at
\(N=3\).

The same pairwise commutativity does not hold in the computed \(N=4\)
examples.  In the numerical \((5,4)\), \(R=4\) example described in
Section~\ref{sec:casimir-example}, the two 20-dimensional blocks have
the same \(\hat n\) and the same central Casimirs.  In \(\widetilde H\),
their splitting is therefore produced entirely by \(-K_3+K_4\),
equivalently by \(25(-K_3+K_4)\) in \(H\).  The two
operators commute in that sector, but a direct calculation with integer
matrices gives \([K_3,K_4]\ne0\) at \(R=5\).  At \((7,4)\), the same
calculation gives \([K_3,K_4]\ne0\), \([K_3,K_5]\ne0\), and
\([K_4,K_5]\ne0\) in sector \(R=5\).  Here \(R=5\) denotes the full
sector \(V_5\), including both possible occupations of the trace mode.

The ancillary program \texttt{verify\_normal\_ordering.py} constructs
the \(K_k\), checks the operator identities used in Proposition
\ref{prop:su3-commutativity}, and reproduces these commutators using
integer arithmetic.  The nonzero commutators show that the displayed
\(K_k\) cannot be simultaneously diagonalized in the \((5,4)\) and
\((7,4)\) examples.  They do not establish chaos or a failure of
integrability, and they do not exclude another commuting family.

\section{Numerical method and results}
\label{sec:method}
\label{sec:results}

This section explains how the sector maps and Hamiltonian matrices are
constructed.  It also states which conclusions are analytic and which
are inferred from numerical ranks.

\subsection{Computational method and status of the data}
\label{sec:gram}

We work in the ordered monomial basis of each \(V_R\), so the matrices \(Q_R\) are sparse.
Define
\begin{equation}\label{eq:stackedA}
A_R=\begin{bmatrix}Q_R\\[2pt] Q_{R-p}^{\dagger}\end{bmatrix} : V_R\to V_{R+p}\oplus V_{R-p}.
\end{equation}
Then
\begin{equation}\label{eq:gram}
H_R=A_R^{\dagger}A_R,
\end{equation}
so \(H_R\) and \(A_RA_R^{\dagger}\) have the same nonzero eigenvalues.
It is therefore more efficient to diagonalize \(A_RA_R^{\dagger}\) whenever
\(\dim V_{R+p}+\dim V_{R-p}<\dim V_R\).
For \(p=5\), \(N=4\), this is what makes the middle sectors feasible.

We imposed three checks on the data:
\begin{enumerate}[nosep]
\item the zero multiplicities in each sector from \eqref{eq:cohom} agree with those from \(H_R\),
\item the data obey the particle--hole symmetry \eqref{eq:phsym},
\item for the full spectra with \(N\le 4\), the energy multiplicities in each sector sum to \(\dim V_R\).
\end{enumerate}

\paragraph{Scope of the numerical results.}
The results below concern only the stated pairs \((p,N)\) and fall into
three categories.
\begin{enumerate}[label=(\roman*),nosep]
\item At the four complete cases \((5,3)\), \((5,4)\), \((5,5)\), \((7,4)\)
we report a numerical BPS multiplicity
\(h_R\) for every~\(R\).
\item The ranks in each sector for all four cases are supplied for every
\(R\) in the ancillary file \texttt{rank\_data.json}.
For \((5,4)\), the spectrum of \(H\) is tabulated through half filling in
Appendix~\ref{app:data54}.  The remaining sectors follow by
particle--hole conjugation.  For \((5,3)\) and \((7,4)\),
we record the distinct energies together with the BPS multiplicities,
with the complete spectra used only for the
consistency checks listed above.
At \((5,5)\) only the multiplicities at zero energy were computed, since the
nonzero levels would require diagonalizing the middle sectors.
\item At \((7,5)\) the data are partial.  Only the first sectors
\(R\le 6\) were computed, giving \(q_{\min}\) but not the full spectrum.
\end{enumerate}

The multiplicities are inferred from numerical ranks of explicitly
constructed integer sector maps.  For \(N\le4\), the zero and nonzero
clusters of singular values are separated by more than four orders of
magnitude.  In the \(N=5\) block computations, the corresponding
eigenvalue clusters of the Gram matrices are separated by more than two orders
of magnitude.  The \(N=5\) results were also reproduced by the second
implementation described in Appendix~\hyperref[app:data]{B}.

As an independent check on the calculation with floating-point numbers,
the ranks at \((5,3)\), \((5,4)\), and \((7,4)\) were recomputed from the
integer matrices of the reduced differential \(\widetilde Q_p\).  Exact
linear algebra was performed over the fields with characteristics
\(2^{31}-1\), \(2147483629\), and \(p\).
All three modular ranks agree with one another and with the numerical
rank in every sector.  This is an exact check over those fields, but
agreement over finitely many fields does not prove the rank over
\(\mathbf Q\).  Accordingly, the positive BPS multiplicities below
remain numerical unless stated otherwise.
For an integer matrix, reduction modulo a prime cannot increase its rank.
Thus a calculation with full column rank over one finite field proves
injectivity over \(\mathbf Q\).  Here \(\widetilde Q_p\) over a finite
field means the reduction of its integral sector matrices, not division
by \(p\) in that field.  If the cohomology in sector \(R\) vanishes over
one finite field, the incoming and outgoing modular ranks add to
\(\dim V_R\).  The corresponding ranks over \(\mathbf Q\) are at least as large,
while \(\widetilde Q_p^2=0\) prevents their sum from exceeding
\(\dim V_R\).  Equality follows, so the cohomology over \(\mathbf Q\)
also vanishes.  No such modular certification is presently claimed for
the \(N=5\) sectors.
The ancillary file \texttt{rank\_data.json} lists the reported ranks,
sector dimensions, nullities, and cohomology dimensions.  The bundle
also contains the program for the modular rank checks, the program that
checks the normal ordering, and the scripts that generate the figures.

We now present the numerical BPS generating functions for the quintic
model at \(N=3,4,5\) and the \(p=7\) model at \(N=4\).  We also give the
case \((7,3)\), where \(Q_7=0\), and partial data at \((7,5)\).
We then use projections to compare the cohomology at different values
of \(N\).

\subsection{The quintic model: \texorpdfstring{$p=5$}{p=5}}

The first nontrivial quintic case is \(N=3\), where the Fock space is
small enough for complete diagonalization.  At \(N=4\), Gram matrices
are used in the middle sectors.

\subsubsection*{\texorpdfstring{$N=3$}{N=3}}

The Hilbert space has dimension \(2^9=512\).
The numerical spectrum contains three distinct energies,
\[
E\in\{0,225,600\},
\]
and the BPS sector has dimension \(440\).
The multiplicities at zero energy, listed by fermion number, are
\begin{equation}\label{eq:h53}
(h_0,h_1,\ldots,h_9)=(0,0,20,75,125,125,75,20,0,0).
\end{equation}
In particular, the first nonzero BPS sector is \(R=2\): \(h_0=h_1=0\).

\subsubsection*{\texorpdfstring{$N=4$}{N=4}}

The Hilbert space has dimension \(2^{16}=65{,}536\).
The numerical spectrum contains nine distinct energies,
\[
E\in\{0,225,600,900,1100,1400,2100,2400,3600\},
\]
and the total BPS multiplicity is
\begin{equation}\label{eq:total44000}
\sum_R h_R=44{,}000.
\end{equation}
The multiplicities at zero energy are
\begin{equation}\label{eq:h54}
(h_0,h_1,\ldots,h_{16})=(0,0,0,0,125,1375,5000,9625,11750,9625,5000,1375,125,0,0,0,0).
\end{equation}
Equivalently, the BPS spectrum begins at fermion number \(R=4\).
The charge profiles for the four complete nontrivial cases are compared in Figure~\ref{fig:profiles}.

\paragraph{Equivalent representations at different energies.}
\label{sec:casimir-example}
In sector \(R=4\), consider the subspace with an unoccupied trace mode.
Writing \(\mathcal E_E\) for the eigenspace with energy \(E\) in this
subspace, the numerical \(\SU(4)\) decompositions are
\[
\mathcal E_0=[0,4,0]\oplus[0,2,0],
\qquad
\mathcal E_{2100}=[1,0,1]\oplus[0,2,0].
\]
The two eigenspaces have dimensions \(125\) and \(35\), respectively.
Each contains one copy of the same 20-dimensional irreducible
representation \([0,2,0]\).  Since these two summands are equivalent,
they have identical eigenvalues for every central
Casimir~\cite{FultonHarris1991}, but they occur at different energies.
No function of the central Casimirs can distinguish the two summands.
Section~\ref{sec:normalH} constructs a commuting family at \((5,3)\).
Whether the model more generally admits a family of pairwise
commuting \(\SU(N)\)-invariant conserved operators, beyond functions of
\(H\), that is defined for all \(N\ge3\) and resolves such
multiplicities remains open.

\begin{numericalresult}[Quintic BPS generating polynomials]
\label{prop:p5}
The numerical BPS generating polynomials are
\begin{align}
\Zbps^{(5,3)}(x)&=5x^2(1+x)^3\bigl(4+3x+4x^2\bigr),\label{eq:Z53-new}\\[3pt]
\Zbps^{(5,4)}(x)&=125x^4(1+x)^4\bigl(1+7x+6x^2+7x^3+x^4\bigr).\label{eq:Z54-new}
\end{align}
\end{numericalresult}

Evaluating at \(x=1\) gives \(440\) and \(44{,}000\), respectively.
After the displayed factors are removed, particle--hole symmetry gives
the reciprocal polynomials
\[
T_3^{(5)}(x)=4+3x+4x^2,
\qquad
T_4^{(5)}(x)=1+7x+6x^2+7x^3+x^4.
\]

\subsection{The \texorpdfstring{$p=7$}{p=7} model}

At \(N=3\), one has \(Q_7=0\) because \(7>2N-1\).  At
\(N=4\), it provides a nontrivial calculation at a different value of
\(p\).

\subsubsection*{\texorpdfstring{$N=3$}{N=3}}

Since \(p=7>2N-1=5\), the trace in the supercharge vanishes.
Accordingly \(Q_7\equiv 0\), the Hamiltonian vanishes identically, and the full Fock space is BPS:
\begin{equation}\label{eq:Z73-trivial}
\Zbps^{(7,3)}(x)=(1+x)^9.
\end{equation}
This agrees with the condition~\eqref{eq:range}.

\subsubsection*{\texorpdfstring{$N=4$}{N=4}}

At \(N=4\), the Hilbert space again has dimension \(65{,}536\).
The numerical spectrum contains eight distinct energies,
\[
E\in\{0,1764,4410,7350,7840,14700,18816,35280\},
\]
and the total BPS multiplicity is
\begin{equation}\label{eq:total59136}
\sum_R h_R=59{,}136.
\end{equation}
The multiplicities at zero energy, listed by fermion number, are
\begin{equation}\label{eq:h74}
(h_0,h_1,\ldots,h_{16})=(0,0,0,70,847,3395,7518,11319,12838,11319,7518,3395,847,70,0,0,0).
\end{equation}
Here the first nonzero BPS sector is \(R=3\).

\begin{numericalresult}[BPS generating polynomial for \(p=7\)]
\label{prop:p7}
For \(p=7\), \(N=4\), the numerical BPS generating polynomial is
\begin{equation}\label{eq:Z74-new}
\Zbps^{(7,4)}(x)=7x^3(1+x)^4\bigl(10+81x+101x^2+144x^3+101x^4+81x^5+10x^6\bigr).
\end{equation}
\end{numericalresult}

The polynomial
\[
T_4^{(7)}(x)=10+81x+101x^2+144x^3+101x^4+81x^5+10x^6
\]
is reciprocal, as follows from particle--hole symmetry.
The BPS fraction is
\(59{,}136/65{,}536\approx 90.2\%\),
which is higher than in the quintic model at the same matrix size.

Figure~\ref{fig:energy} compares the distinct energy levels in the two
computed \(N=4\) models.  The representation comparison in the quintic
case shows that the quadratic Casimir does not determine all these
levels.

\begin{figure}[H]
\centering
\includegraphics[width=0.55\textwidth]{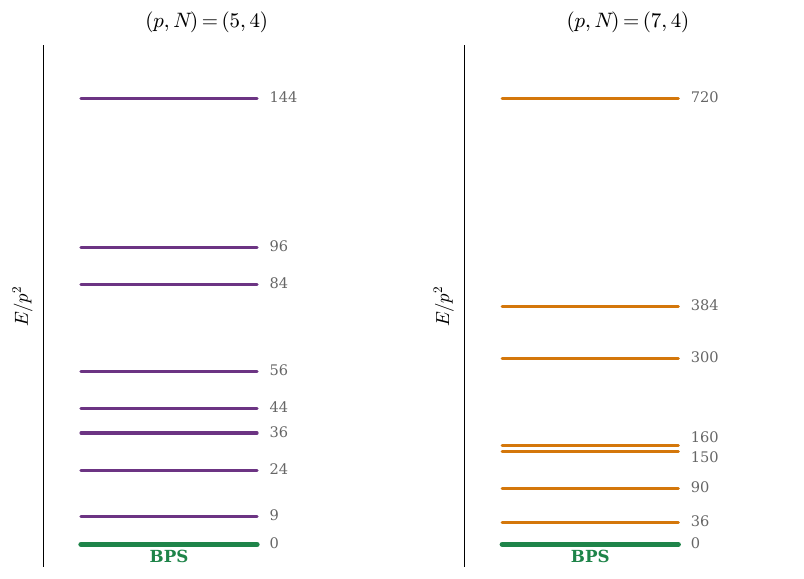}
\caption{Distinct numerical energy levels in the calculated models.
Each horizontal line marks a value of \(E/p^2\), and the green line is
the BPS value \(E=0\).  Equivalent irreducible representations can
occur at different energies, so the quadratic Casimir does not
determine these levels.}
\label{fig:energy}
\end{figure}

\subsection{Rank sequences and the \texorpdfstring{$N=5$}{N=5} data}
\label{sec:partial}

The cohomology formula~\eqref{eq:cohom} determines \(h_R\) from the ranks of the rectangular maps \(Q_R\).
For \(p=5\), \(N=4\), the reported rank sequence is
\begin{equation}\label{eq:rankseq}
(r_0,\ldots,r_{11})=(1,16,120,560,1695,2992,2992,1695,560,120,16,1),
\end{equation}
which is palindromic (as guaranteed by Theorem~\ref{thm:rankpal}) and unimodal.
The computation over finite fields certifies injectivity for \(R\le3\), so
there are no BPS states at those charges.  Within the reported
numerical rank assignment, the first rank deficiency occurs at \(R=4\),
where \(\binom{16}{4}-r_4=125\).  This reported kernel dimension equals the power
\(5^3\) that will appear in the factorization of the BPS polynomial in
Section~\ref{sec:factorization}.  The equality is not general: at
\((5,3)\), the first BPS multiplicity is \(h_2=20\), whereas the common
power dividing the reduced coefficients is \(5\).

At \(N=5\), the sector maps were computed using sparse matrix assembly
and block decomposition of the Gram matrices \(Q_R^\top Q_R\).
Computational details are given in Appendix~\hyperref[app:data]{B}.
The results are summarized in Table~\ref{tab:55data}.

\begin{table}[ht]
\centering
\small
\begin{tabular}{c r r r r}
\toprule
\(R\) & \(\dim V_R\) & \(\dim\ker Q_R\) & \(\mathrm{rank}\,Q_{R-5}\) & \(h_R\) \\
\midrule
0--4 & \(\le 12{,}650\) & 0 & --- & 0 \\
5 & 53{,}130 & 1 & 1 & 0 \\
6 & 177{,}100 & 25 & 25 & 0 \\
7 & 480{,}700 & 300 & 300 & 0 \\
8 & 1{,}081{,}575 & 36{,}050 & 2{,}300 & 33{,}750 \\
9 & 2{,}042{,}975 & 319{,}525 & 12{,}650 & 306{,}875 \\
10 & 3{,}268{,}760 & 1{,}221{,}254 & 53{,}129 & 1{,}168{,}125 \\
\bottomrule
\end{tabular}
\caption{Numerical BPS data at low fermion number for \((p,N)=(5,5)\).
For \(R\le 4\), the reported rank is full.  The sectors \(R=11,12\) follow
from the symmetry \(r_R=r_{N^2-p-R}\), and the remaining sectors follow from
particle--hole symmetry.}
\label{tab:55data}
\end{table}

\noindent Thus \(q_{\min}=8\) at \((p,N)=(5,5)\).
This is one above the empirical \(N\le4\) extrapolation
stated in Section~\ref{sec:factorization}.  Figure~\ref{fig:onset}
shows how the cohomology first becomes nonzero.
The first two nonzero BPS multiplicities are
\begin{equation}\label{eq:h8_55}
h_8^{(5,5)}=33{,}750=2\cdot 3^3\cdot 5^4,
\qquad
h_9^{(5,5)}=306{,}875=5^4\cdot 491.
\end{equation}
Both reported values are divisible by \(5^4\) but not \(5^5\), and their greatest common
divisor is exactly~\(5^4=625\).
By particle--hole symmetry, \(h_{17}^{(5,5)}=33{,}750\) and \(h_{16}^{(5,5)}=306{,}875\).

For the first two surviving sectors (\(R=8,9\)), the individual nullities and incoming ranks are each divisible by \(5^2\) but not~\(5^4\).
For instance,
\[
\dim\ker Q_8=36{,}050=5^2\cdot 1{,}442,
\qquad
\mathrm{rank}\,Q_3=2{,}300=5^2\cdot 92,
\]
but \(h_8=5^4\cdot 54\) because \(1{,}442-92=1{,}350=5^2\cdot 54\).
At \(R=10\), neither individual term is divisible by~\(5^2\), whereas their difference is divisible by~\(5^4\): \(h_{10}=1{,}168{,}125=5^4\cdot 1{,}869\).
Thus the extra \(5^4\) divisibility appears only after subtracting the incoming rank from the kernel dimension:
\(\dim\ker Q_R\equiv\mathrm{rank}\,Q_{R-5}\pmod{5^4}\).
It is not divisibility of either term separately by a higher power of \(p\).

The numerical calculation at \(R=10\) gives
\(r_{10}=2{,}047{,}506\) and
\(h_{10}=1{,}168{,}125\).
Together with the symmetry of the ranks and particle--hole symmetry, these
rank assignments give the following numerical generating polynomial:
\begin{equation}\label{eq:Z55-predicted}
\begin{split}
\Zbps^{(5,5)}(x)={}&33{,}750x^8+306{,}875x^9
+1{,}168{,}125x^{10}+2{,}556{,}875x^{11}\\
&+3{,}674{,}375x^{12}+3{,}674{,}375x^{13}
+2{,}556{,}875x^{14}\\
&+1{,}168{,}125x^{15}+306{,}875x^{16}
+33{,}750x^{17}.
\end{split}
\end{equation}
The corresponding numerical total is \(15{,}480{,}000\).
Section~\ref{sec:factorization} shows that the same polynomial contains
\((1+x)^5\) and that \(r_{10}\) can instead be determined from the
ranks at lower fermion number and the exact relation
\((1+x)\mid\mathcal R_{5,5}(x)\).

\begin{figure}[p]
\centering
\includegraphics[width=0.98\textwidth]{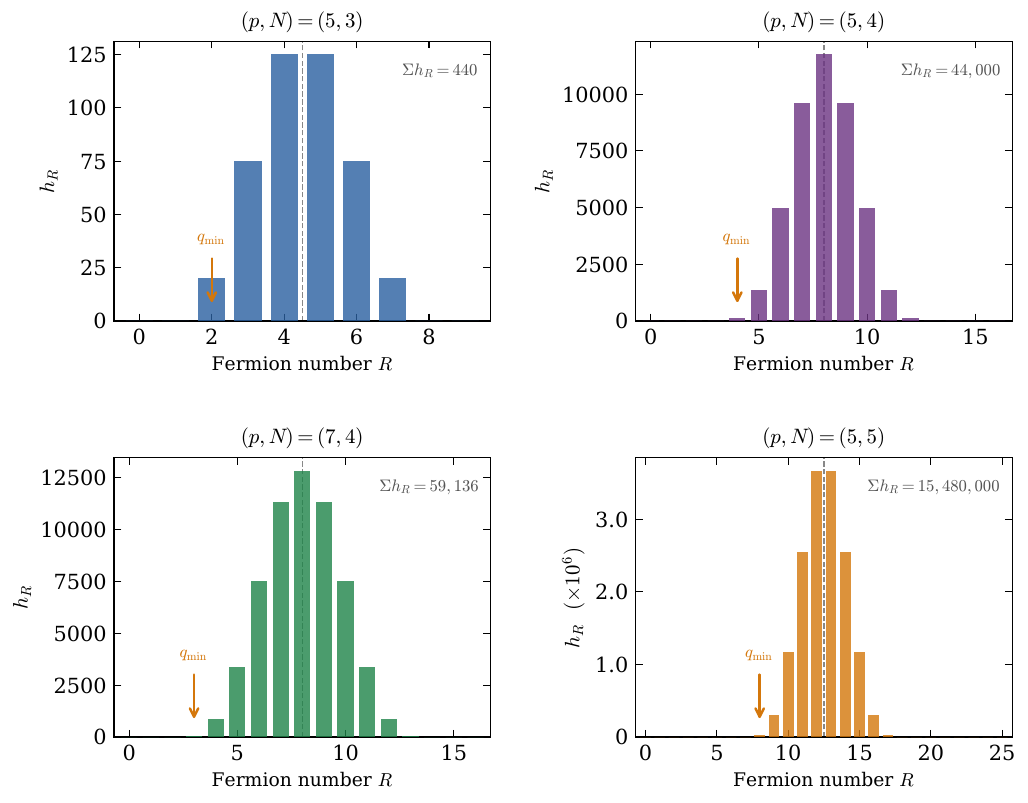}
\caption{Computed BPS charge profiles.
Numerical BPS multiplicities \(h_R\), resolved by fermion number, for
the four nontrivial cases with data in every sector.
The dashed line marks the particle--hole symmetry axis at half filling (\(R=N^2/2\)).
Arrows indicate the minimum BPS charge \(q_{\min}\).  At \((5,5)\) it is
one above the empirical \(N\le4\) extrapolation in
\eqref{eq:smallN-patterns}.}
\label{fig:profiles}
\end{figure}

\begin{figure}[H]
\centering
\includegraphics[width=0.66\textwidth]{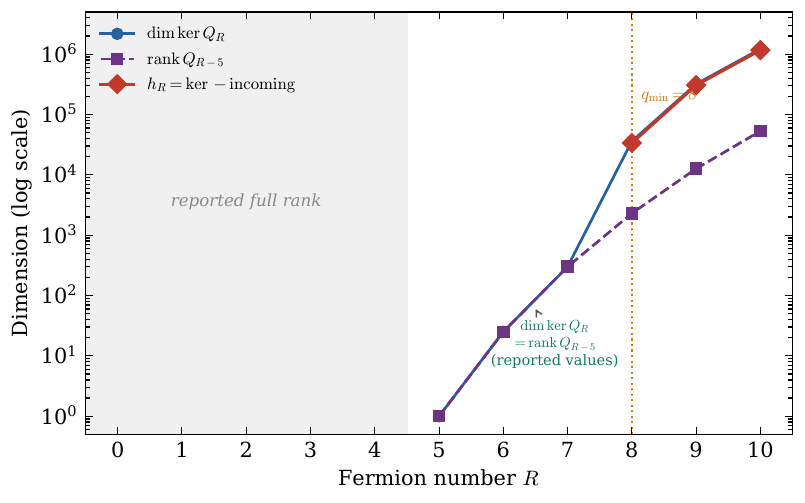}
\caption{For the reported numerical ranks, the figure shows the onset
of BPS cohomology at \((p,N)=(5,5)\).
For \(R\le 4\), \(Q_R\) is injective and no kernel exists.
At \(R=5,6,7\), a kernel appears in blue, but its dimension equals the
incoming rank in purple, so \(h_R=0\).
Cohomology first survives at \(R=8=q_{\min}\), where the kernel
dimension exceeds the incoming rank.
The red diamonds plot
\(h_R=\dim\ker Q_R-\operatorname{rank}Q_{R-5}\).}
\label{fig:onset}
\end{figure}

For \(p=7\), \(N=5\), the same numerical method gives full column rank for \(Q_R\) at \(R=0,\ldots,5\) and
\begin{equation}\label{eq:h6_75}
h_6^{(7,5)}=2{,}352=2^4\cdot 3\cdot 7^2
\end{equation}
(with no incoming map since \(R-p<0\)), so \(q_{\min}=6\).
This is again one above the empirical extrapolation from \(N\le4\)
stated in Section~\ref{sec:factorization}.
The nonzero entries of the sector matrices are \(\pm 7,\pm 14,\pm 42\),
consistent with the factor of seven proved in
Proposition~\ref{prop:cyclic}.

\subsection{Projections between ranks and fortuity}

To determine which computed BPS classes can come from larger matrices,
we use projection rather than inclusion.
In the nontrivial comparisons considered here, the natural inclusion
that retains the first \(N\) rows and columns,
\(\mathcal H_N\hookrightarrow \mathcal H_{N+1}\), is generally not a
cochain map.
Even on the vacuum, \(Q_p^{(N+1)}\) creates terms involving the new index that are absent from the image of \(Q_p^{(N)}\).
The useful cochain map therefore goes in the opposite direction.

Write
\[
W_N:=\mathrm{span}\{\Psi_{ij}\}_{1\le i,j\le N},
\qquad
\mathcal H_N=\Lambda^\bullet W_N,
\qquad
\mathcal B_R^{(p,N)}:=\ker Q_R/\operatorname{im}Q_{R-p}.
\]
For \(M>N\), let
\[
\pi_{M\to N}:\mathcal H_M\to \mathcal H_N
\]
be the linear map that deletes every basis monomial containing any generator \(\Psi_{ij}\) with \(i>N\) or \(j>N\).
Equivalently, it is the projection onto the subspace generated by monomials using only the first \(N\) rows and columns.
It preserves fermion number.

\begin{proposition}[Projection is a cochain map]
\label{thm:rankproj}
For every odd $p$ and every $M>N$, the projection
$\pi_{M\to N}:\mathcal H_M\to\mathcal H_N$
that sends each $\Psi_{ij}$ to itself if $i,j\le N$ and to zero otherwise,
extended multiplicatively to the exterior algebra, satisfies
\[
\pi_{M\to N}Q_p^{(M)}=Q_p^{(N)}\pi_{M\to N}.
\]
Hence $\pi_{M\to N}$ induces well defined linear maps on cohomology
\[
\Pi_{M\to N}^{(R)}:\mathcal B_R^{(p,M)}\to \mathcal B_R^{(p,N)},
\]
and for $L>M>N$ these maps satisfy
\(\Pi_{L\to N}^{(R)}=
\Pi_{M\to N}^{(R)}\circ\Pi_{L\to M}^{(R)}\).
The \emph{monotone BPS subspace} at rank~$N$ is
\[
\mathcal B_{R,\mathrm{mon}}^{(p,N)}:=\bigcap_{M>N}\operatorname{Im}\,\Pi_{M\to N}^{(R)}.
\]
A BPS class is \emph{fortuitous} in the present family if its coset in
$\mathcal B_R^{(p,N)}/\mathcal B_{R,\mathrm{mon}}^{(p,N)}$ is nonzero.
Equivalently, it lies outside $\mathcal B_{R,\mathrm{mon}}^{(p,N)}$.
Hence, for some $M_0>N$, it is not in
$\operatorname{Im}\,\Pi_{M_0\to N}^{(R)}$.  This is the definition used
below.  Its relation to the covering construction of
Ref.~\cite{Chang2024} is explained in Appendix~\ref{app:projections}.
\end{proposition}

Because the maps compose, the images being intersected can only decrease
as \(M\) increases.  Their intersection is the part of the cohomology at
\(N\) that lies in the image from every larger rank.

The proof is deferred to Appendix~\ref{app:projections}.

\begin{remark}
In the D1--D5 system~\cite{ChangLinZhang2025}, the analogous projection does
not commute with the supercharge.  A preimage of a closed class can then
fail to be closed or can become exact at the larger rank.  That
complication does not arise here.
\end{remark}

With respect to the standard embedding \(\SU(N)\subset \SU(M)\), the
map \(\pi_{M\to N}\) intertwines the \(\SU(N)\) actions.  Once explicit
\(\SU(N)\) decompositions are available, the comparison can be made
separately for each irreducible representation.

We denote by
$\mathcal B^{(p,N)}_{R,\mathrm{fort}}:=
\mathcal B_R^{(p,N)}/\mathcal B^{(p,N)}_{R,\mathrm{mon}}$
the quotient that measures the independent BPS directions not
contained in the monotone subspace.

The reported numerical ranks at \(N=5\) (Table~\ref{tab:55data}) give
\(h_R^{(5,5)}=0\) for \(R=0,\ldots,7\).
The induced map \(\Pi_{5\to4}^{(R)}\) maps the \(N=5\) BPS cohomology
in sector \(R\) to the \(N=4\) BPS cohomology in the same sector, and
the induced maps compose.
If the sector at larger \(N\) is empty, its image is zero.
Conditional on that numerical vanishing, the \(N=4\) BPS classes in that
sector are therefore fortuitous under the present definition.
Figure~\ref{fig:tower} illustrates this criterion.

\begin{figure}[H]
\centering
\includegraphics[width=0.85\textwidth]{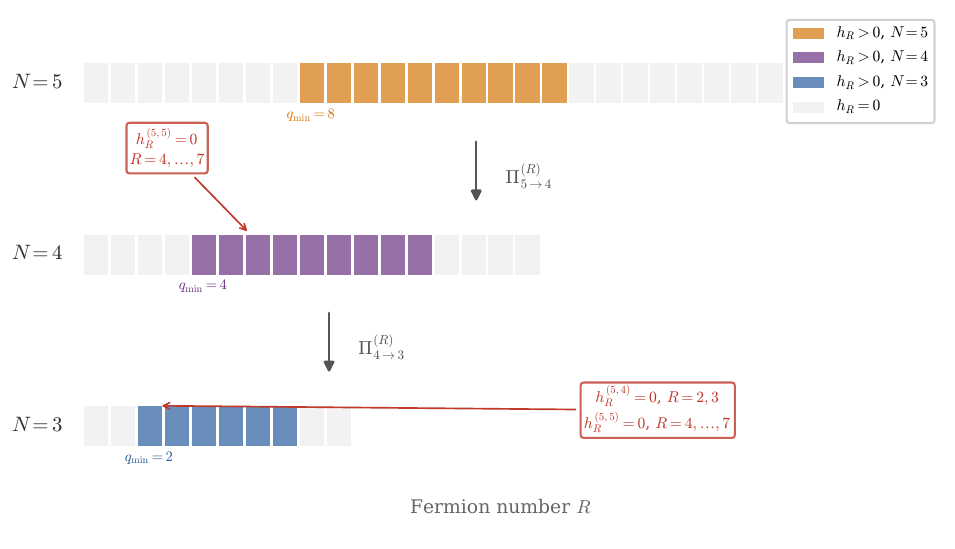}
\caption{Projection of BPS cohomology from larger to smaller \(N\) in
the quintic model.  Each row shows the sectors at fixed fermion number
\(R\).  Coloured cells have \(h_R>0\), and the maps
\(\Pi_{5\to4}^{(R)}\) and \(\Pi_{4\to3}^{(R)}\) preserve \(R\).  If the BPS space
at a larger value of \(N\) vanishes, its image at a smaller value of
\(N\) is zero.  Conditional on the displayed numerical vanishings,
every nonzero BPS class at \((5,3)\), and every nonzero BPS class with
\(R=4,\ldots,7\) at \((5,4)\), is fortuitous.}
\label{fig:tower}
\end{figure}

\begin{corollary}[Vanishing criterion for fortuity]
\label{prop:firstfort}
If $\mathcal B_R^{(p,M_0)}=0$ for some $M_0>N$, then
$\mathcal B_{R,\mathrm{mon}}^{(p,N)}=0$.
Hence every nonzero BPS class in $\mathcal B_R^{(p,N)}$ is fortuitous.
\end{corollary}

\begin{proof}
If $\mathcal B_R^{(p,M_0)}=0$, then
$\Pi_{M_0\to N}^{(R)}=0$.  By transitivity, every map from a larger rank
$L>M_0$ factors through that empty sector:
\[
\Pi_{L\to N}^{(R)}
=\Pi_{M_0\to N}^{(R)}\circ\Pi_{L\to M_0}^{(R)}=0.
\]
Hence
$\mathcal B_{R,\mathrm{mon}}^{(p,N)}
=\bigcap_{L>N}\operatorname{Im}\Pi_{L\to N}^{(R)}=0$.
\end{proof}

For the first row of Table~\ref{tab:fortuity}, the required vanishings
occur at \(N=4\) for \(R=2,3\) and at \(N=5\) for \(R=4,\ldots,7\).
The second row uses the \(N=5\) quintic vanishings for
\(R=4,\ldots,7\).  The third uses the computed \(p=7\) vanishings at
\(N=5\) for \(R\le5\).  Their numerical status is described in
Section~\ref{sec:gram}.

\begin{table}[ht]
\centering
\small
\begin{tabular}{ccrrrl}
\toprule
\(p\) & \(N\) & identified & total BPS & fraction & sectors \\
\midrule
5 & 3 & 440 & 440 & 100\% & all (\(R=2,\ldots,7\)) \\
5 & 4 & 16{,}125 & 44{,}000 & 37\% & \(R=4,\ldots,7\) \\
7 & 4 & 4{,}312 & 59{,}136 & 7\% & \(R=3,4,5\)\\
\bottomrule
\end{tabular}
\caption{BPS states identified as fortuitous by projection at fixed
fermion number, conditional on the reported numerical vanishings at
the larger rank.  Particle--hole conjugation gives equal counts in a
separate comparison at fixed hole number.  Those counts are not added
here.}
\label{tab:fortuity}
\end{table}

At fixed hole number \(s\), particle--hole conjugation compares the
sectors \(R_M=M^2-s\) and \(R_N=N^2-s\) at ranks \(M\) and \(N\).
This differs from the projection at fixed fermion number, which
compares the same value of \(R\) at both ranks.  A conclusion from the
fixed hole comparison therefore does not determine the stable image at
fixed fermion number.

\paragraph{Index bound and the numerical difference at \texorpdfstring{$(p,N)=(5,4)$}{(p,N)=(5,4)}.}

For \(p=5\), \(N=4\), the Euler characteristics in the five residue
classes are obtained from \eqref{eq:index}:
\begin{equation}\label{eq:I5N4}
(I_0,I_1,I_2,I_3,I_4)=(3625,\,3625,\,-9500,\,11750,\,-9500).
\end{equation}
Therefore the exact lower bound from the indices is
\begin{equation}\label{eq:indexfloorN4}
\sum_{a=0}^{4}|I_a|=38{,}000.
\end{equation}
Comparing with the computed BPS count \eqref{eq:total44000} gives the
difference
\begin{equation}\label{eq:fort6000}
44{,}000-38{,}000=6{,}000.
\end{equation}

The difference is not distributed uniformly across residue classes
(Figure~\ref{fig:fortuity}(b)).
Writing \(Z_a=\sum_{R\equiv a \,(\mathrm{mod}\,5)}h_R\), one finds
\begin{center}
\renewcommand{\arraystretch}{1.15}
\begin{tabular}{cccc}
\toprule
residue \(a\) & \(|I_a|\) & \(Z_a\) & difference \\
\midrule
0 & 3625 & 6375 & 2750 \\
1 & 3625 & 6375 & 2750 \\
2 & 9500 & 9750 & 250 \\
3 & 11750 & 11750 & 0 \\
4 & 9500 & 9750 & 250 \\
\bottomrule
\end{tabular}
\end{center}
In residue~3 the BPS count equals the index bound, while residues~0 and~1
give the largest differences.
This decomposition compares, residue class by residue class, the
computed BPS count with the lower bound supplied by the protected
index.  It does not select individual BPS states.

\begin{figure}[H]
\centering
\includegraphics[width=0.75\textwidth]{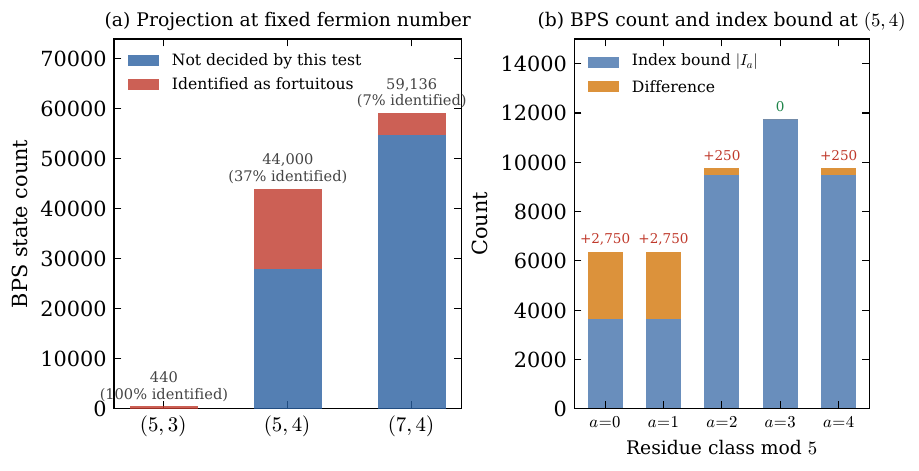}
\caption{Projection test and index comparison.
Left panel: red denotes states conditionally identified as fortuitous by
the projection test, using the reported numerical vanishings at larger
\(N\).  Blue denotes states whose status is not determined by this test.
Percentages refer only to this comparison.
Right panel: decomposition by residue class at \((5,4)\), showing the
lower bound \(|I_a|\) (blue) and the difference from that bound (orange)
in each class \(a\bmod 5\).  The BPS count in residue \(a=3\) equals
the bound.}
\label{fig:fortuity}
\end{figure}

\section{Factorization and map ranks}
\label{sec:factorization}

The BPS generating functions in Section~\ref{sec:results} display a
common factorization pattern.
We now state that pattern precisely and relate it to the ranks of the maps \(Q_R\).
Table~\ref{tab:master} collects the main numerical results for the four
nontrivial cases with data in every sector.

\begin{table}[ht]
\centering
\renewcommand{\arraystretch}{1.3}
\footnotesize
\begin{tabular}{cc rrr c ccc c}
\toprule
$(p,N)$ & $2^{N^2}$ & $\sum h_R$ & index bound & difference
  & $q_{\min}$ & $p^{b_{p,N}}$ & $\deg T$ & $T(1)$
  & $\#\{E/p^2\}$ \\
\midrule
$(5,3)$ & $512$ & $440$ & $400$ & $40$
  & $2$ & $5^1$ & $2$ & $11$
  & $3$ \\
$(5,4)$ & $65{,}536$ & $44{,}000$ & $38{,}000$ & $6{,}000$
  & $4$ & $5^3$ & $4$ & $22$
  & $9$ \\
$(7,4)$ & $65{,}536$ & $59{,}136$ & $55{,}468$ & $3{,}668$
  & $3$ & $7^1$ & $6$ & $528$
  & $8$ \\
$(5,5)$ & $33{,}554{,}432$ & $15{,}480{,}000$ & $11{,}781{,}250$ & $3{,}698{,}750$
  & $8$ & $5^4$ & $4$ & $774$
  & --- \\
\bottomrule
\end{tabular}
\caption{Summary of the numerical results at fixed \(N\).
For each computed case, the columns give the dimension of the Hilbert
space, the total BPS count, the lower bound from the Witten indices and
the difference between them, the minimum BPS charge $q_{\min}$, the power
$p^{b_{p,N}}$, the degree and evaluation $T(1)$
of the final quotient, and the number of distinct rescaled energy
levels $E/p^2$.}
\label{tab:master}
\end{table}

\noindent Whenever \((1+x)^N\) divides \(\Zbps^{(p,N)}(x)\), set
\begin{equation}\label{eq:reducedT}
P_N^{(p)}(x):=x^{-q_{\min}}(1+x)^{-N}\Zbps^{(p,N)}(x).
\end{equation}
Here \(q_{\min}\) is the minimum BPS charge.  If \(P_N^{(p)}\) has
integer coefficients, let \(b_{p,N}\) be the largest integer for which
\(p^{b_{p,N}}\) divides all of them, and define
\[
T_N^{(p)}(x):=p^{-b_{p,N}}P_N^{(p)}(x).
\]

Corollary~\ref{cor:structural-factors} accounts for two factors of
\(1+x\) when \(N\) is even and three when \(N\) is odd.  Thus the
observed factor \((1+x)^N\) contains \(N-2\) further factors when \(N\)
is even and \(N-3\) when \(N\) is odd.  In particular, the factor
\((1+x)^3\) at \((p,N)=(5,3)\) is entirely structural.  The additional
divisibility first appears in the present data at \(N=4\).

Particle--hole conjugation gives
\(\Zbps^{(p,N)}(x)=x^{N^2}\Zbps^{(p,N)}(x^{-1})\).  Therefore, whenever
\(x^{q_{\min}}(1+x)^N\) divides it, the quotient is automatically
palindromic.

\begin{conjecture}[Additional factors]
\label{conj:factor}
For odd \(p\ge5\) with \(p\le2N-1\),
\begin{equation}\label{eq:factor-divisibility}
(1+x)^N\mid\Zbps^{(p,N)}(x)
\qquad\text{in }\Z[x].
\end{equation}
After the minimum charge and this factor are removed, the quotient has
nonnegative coefficients and a common positive factor of \(p\):
\begin{equation}\label{eq:factor-positive-p}
P_N^{(p)}(x)\in p\,\Z_{\ge0}[x].
\end{equation}
Equivalently, \(b_{p,N}\ge1\) and
\begin{equation}\label{eq:factorlaw}
\Zbps^{(p,N)}(x)=p^{b_{p,N}}\,x^{q_{\min}}(1+x)^N\, T_N^{(p)}(x),
\end{equation}
where \(T_N^{(p)}(x)\) has nonnegative integer coefficients and its
coefficients do not all have a further factor of \(p\).  Its
palindromicity follows from
particle--hole conjugation.
\end{conjecture}

For \((p,N)=(5,3),(5,4),(7,4)\), the computed parameters satisfy the
simple formulas
\begin{equation}\label{eq:smallN-patterns}
q_{\min}=\binom{N-m}{2}+m,\qquad
b_{p,N}=\binom{N-m}{2},\qquad m=\frac{p-3}{2}.
\end{equation}
These formulas describe the three numerical data sets with \(N\le4\).
They are not claims for arbitrary \(N\).
These patterns do not extend to \(N=5\): the computations give \(q_{\min}=8\) at \((5,5)\) and \(q_{\min}=6\) at \((7,5)\), each one unit above the \(N\le4\) value.
They also give \(b_{5,5}=4\), rather than the extrapolated value \(6\).
No explicit formula for $q_{\min}$ or $b_{p,N}$ valid for arbitrary~$N$ is currently known.
The numerical polynomial at \((5,5)\) still has the form
\eqref{eq:factorlaw}.
The ten reported ranks below the middle and rank symmetry leave only
the middle rank undetermined.  The exact relation
\((1+x)\mid\mathcal R_{5,5}(x)\) fixes that rank, and the direct
numerical calculation of \(h_{10}\) agrees with
\eqref{eq:Z55-predicted}.  Table~\ref{tab:newdata} compares the two
calculations at \(N=5\).

\begin{table}[ht]
\centering
\renewcommand{\arraystretch}{1.15}
\begin{tabular}{ccccccc}
\toprule
\(p\) & \(N\) & \(q_{\min}\) & \(p^{b_{p,N}}\) & \(\deg T\) & total BPS & sectors computed \\
\midrule
5 & 5 & 8 & \(5^4\) & 4 & 15{,}480{,}000 & all \\
7 & 5 & 6 & --- & --- & --- & \(R\le6\) \\
\bottomrule
\end{tabular}
\caption{Numerical results at \(N=5\).  Both values of \(q_{\min}\)
are one above the formulas inferred from \(N\le4\) in
\eqref{eq:smallN-patterns}.  The dashes in the \(p=7\) row require
computations beyond \(R=6\).}
\label{tab:newdata}
\end{table}

All four numerical polynomials have the factor \((1+x)^N\), a common
positive power of \(p\), and a final quotient with nonnegative
coefficients.  Palindromicity of that quotient follows from
particle--hole conjugation once the divisibility is known.
The sequence \(b_{5,N}=1,3,4\) for \(N=3,4,5\) has no known explicit formula.

After the factors proved in Corollary~\ref{cor:structural-factors} are
removed, the conjectural content is the remaining divisibility needed
to reach \((1+x)^N\), the positive common factor of \(p\), and
nonnegativity of the final quotient.

\begin{theorem}[Identity relating the BPS polynomial to the map ranks]
\label{thm:rankpoly}
Define the polynomial
\begin{equation}\label{eq:rankpoly}
\mathcal R_{p,N}(x):=\sum_{R\ge 0}r_R\,x^R,
\qquad r_R:=\operatorname{rank}Q_R=\operatorname{rank}\widetilde Q_{p,R}.
\end{equation}
Then
\begin{equation}\label{eq:Zbps-rank}
\Zbps^{(p,N)}(x)=(1+x)^{N^2}-(1+x^p)\,\mathcal R_{p,N}(x),
\end{equation}
and consequently, for every $1\le k\le N^2$,
\begin{equation}\label{eq:rank-equiv}
(1+x)^k \mid \Zbps^{(p,N)}(x)
\;\Longleftrightarrow\;
(1+x)^{k-1}\mid \mathcal R_{p,N}(x).
\end{equation}
In particular, $(1+x)^N\mid\Zbps^{(p,N)}(x)$ if and only if
$(1+x)^{N-1}\mid\mathcal R_{p,N}$.
Thus the factor \((1+x)^N\) is equivalent to a divisibility property of the ranks \(r_R\).
Any explanation of the factor \((1+x)^N\) should therefore also account
for the corresponding divisibility of \(\mathcal R_{p,N}(x)\).
\end{theorem}

\begin{proof}
Equation~\eqref{eq:Zbps-rank} follows from~\eqref{eq:cohom} by summing
$h_R=\dim V_R-r_R-r_{R-p}$ against~$x^R$ and using
$\sum_R\binom{N^2}{R}x^R=(1+x)^{N^2}$.
For~\eqref{eq:rank-equiv}, since \(1\le k\le N^2\), the term
\((1+x)^{N^2}\) in~\eqref{eq:Zbps-rank} is divisible by
\((1+x)^k\).  Hence
\[
  (1+x)^k\mid\Zbps^{(p,N)}(x)
  \;\Longleftrightarrow\;
  (1+x)^k\mid(1+x^p)\mathcal R_{p,N}(x).
\]
For odd \(p\),
\[
  1+x^p=(1+x)(1-x+x^2-\cdots+x^{p-1}),
\]
and the second factor takes the nonzero value \(p\) at \(x=-1\).
It is therefore coprime to \(1+x\) over \(\mathbb Q[x]\), so
\[
  (1+x)^k\mid(1+x^p)\mathcal R_{p,N}(x)
  \;\iff\;
  (1+x)^{k-1}\mid\mathcal R_{p,N}(x).
\]
Because powers of \(1+x\) are monic, the same divisibility statements
hold over \(\mathbb Z[x]\).  This proves~\eqref{eq:rank-equiv}.
\end{proof}

The reported rank data exhibit this stronger divisibility.  The
polynomials \(\mathcal R_{p,N}\) factor as
\begin{align}
\mathcal R_{5,3}(x)&=(1+x)^2\bigl(1+7x+x^2\bigr),\label{eq:J53}\\[3pt]
\mathcal R_{5,4}(x)&=(1+x)^3\bigl(1+13x+78x^2+286x^3+590x^4+286x^5+78x^6+13x^7+x^8\bigr),\label{eq:J54}\\[3pt]
\mathcal R_{7,4}(x)&=(1+x)^3\bigl(1+13x+78x^2+216x^3+78x^4+13x^5+x^6\bigr),\label{eq:J74}
\end{align}
The parenthetical factor in each line has positive integer
coefficients.
Theorem~\ref{thm:rankpal} proves the palindromicity of \(\mathcal R_{p,N}\).
Since \(\mathcal R_{p,N}\) and \((1+x)^{N-1}\) are palindromic, the
quotient is automatically palindromic whenever the divisibility holds.
The exterior algebra of creation operators acts by maps of the complex.
Every creation operator \(c=\Psi_{ij}\) supercommutes with \(Q_p\):
\begin{equation}\label{eq:supercommute}
c\,Q_p=-Q_p\,c,
\end{equation}
since both are odd and built purely from creators.
Hence left multiplication by \(c\) sends \(\operatorname{im}Q_R\) into \(\operatorname{im}Q_{R+1}\) and \(\ker Q_R\) into \(\ker Q_{R+1}\), and therefore acts on cohomology.
Section~\ref{sec:discussion} considers whether the \(N-1\) traceless
diagonal creation modes can explain the remaining factors of \(1+x\).

\begin{proposition}[Determination of the middle rank at \((5,5)\)]
\label{prop:cascade}
Conditional on the ten reported numerical ranks \(r_0,\ldots,r_9\),
the proved rank symmetry and the exact relation
\((1+x)\mid\mathcal R_{5,5}(x)\) determine the middle rank~\(r_{10}\).
For the reconstructed numerical polynomial \(\mathcal R_{5,5}\), the
reported ranks satisfy the additional relation that gives
\((1+x)^4\mid\mathcal R_{5,5}\), while
\((1+x)^5\nmid\mathcal R_{5,5}\).
\end{proposition}

\begin{proof}
Let \(P(x)=x^{2d}P(1/x)\) be palindromic of even degree~\(2d\).
If \(P(-1)=0\), then \(P/(1+x)\) is palindromic of odd degree and
therefore also vanishes at \(-1\).  Thus a first root at \(-1\) has
multiplicity at least two.

At \((p,N)=(5,5)\), the polynomial \(\mathcal R_{5,5}\) has degree~\(20\).
Corollary~\ref{cor:structural-factors} gives the exact condition
\(\mathcal R_{5,5}(-1)=0\).  Rank symmetry expresses every coefficient
except \(r_{10}\) in terms of \(r_0,\ldots,r_9\).  The exact condition
is therefore one linear equation for \(r_{10}\), and it gives
\[
r_{10}=2{,}047{,}506.
\]
The reconstructed value agrees with the direct numerical calculation
at \(R=10\) reported in Section~\ref{sec:partial}.
Palindromicity then supplies the second factor of \(1+x\).

The quotient \(\mathcal R_{5,5}/(1+x)^2\) is palindromic of even
degree~18.  The reported ranks satisfy the additional condition that
this quotient vanishes at \(-1\), and palindromicity supplies the
fourth factor.  The next obstruction is
\[
\left.\frac{\mathcal R_{5,5}(x)}{(1+x)^4}\right|_{x=-1}
=13{,}750=2\cdot5^4\cdot11\ne0.
\]
\end{proof}

For the reconstructed numerical polynomial \(\mathcal R_{5,5}\),
\[
\mathcal R'_{5,5}(-1)=\mathcal R''_{5,5}(-1)
=\mathcal R'''_{5,5}(-1)=0,\qquad
\mathcal R^{(4)}_{5,5}(-1)=330{,}000\ne0.
\]
Theorem~\ref{thm:rankpoly} gives the factor \((1+x)^5\) in the
reconstructed numerical BPS polynomial~\eqref{eq:Z55-predicted}, which
can be written as
\begin{equation}\label{eq:Z55-factorized}
\Zbps^{(5,5)}(x)
=5^4x^8(1+x)^5
\bigl(54+221x+224x^2+221x^3+54x^4\bigr).
\end{equation}
Thus the reconstruction uses the ten numerical ranks below the middle,
the proved symmetry of the ranks, and the exact relation
\((1+x)\mid\mathcal R_{5,5}(x)\).  It does
not assume the further factors observed in the data.  The direct
calculation at \(R=10\) supplies an independent numerical check.

\begin{remark}[Relation to the cubic model]
\label{rem:cubic}
Reference~\cite{Chen2025} found the exact total generating function
\begin{equation}\label{eq:cubicbenchmark}
\Zbps^{(3,N)}(x)
=3^{\binom N2}x^{\binom N2}(1+x)^N.
\end{equation}
Thus the factor \((1+x)^N\) and a positive power of the supercharge
degree already occur for \(p=3\), with
\(b_{3,N}=q_{\min}=\binom N2\) and \(T=1\).  For \(p\ge5\), the
question is whether this structure persists when the Hamiltonian is no
longer determined in general by the quadratic Casimir.
\end{remark}

\section{Discussion and open questions}
\label{sec:discussion}

The numerical generating polynomials in
Eqs.~\eqref{eq:Z53-new}, \eqref{eq:Z54-new}, \eqref{eq:Z74-new}, and
\eqref{eq:Z55-predicted} are the main computational results.  All four
contain \((1+x)^N\).  The data at \(N=5\) also show that the formulas
for \(q_{\min}\) and \(b_{p,N}\) inferred from \(N\le4\) do not continue
unchanged.  The factorization conjecture remains consistent with the
four polynomials, but its parameters are not known in general.

Corollary~\ref{cor:structural-factors} proves divisibility by \((1+x)^2\)
for every \(N\ge2\), and by \((1+x)^3\) when \(N\) is odd.
Theorem~\ref{thm:rankpal} gives the rank symmetry for odd
\(p\le2N-1\).  Theorem~\ref{thm:rankpoly} expresses the remaining
divisibility as a condition on \(\mathcal R_{p,N}\), but these results do
not prove \((1+x)^N\) for arbitrary \(N\).

The Hamiltonian analysis addresses a different question.  In the
numerical \((5,4)\), \(R=4\) example, equivalent \(\SU(4)\) summands occur
at different energies.  The exact normal ordering identities show that
the combination \(-K_3+K_4\) in \(\widetilde H=H/25\) produces this
splitting.  This result concerns the energy spectrum, whereas
\((1+x)^N\) concerns multiplicities at zero energy.  The formulas for
the operators \(K_k\) and their commutators do not explain that factor.

The projections between matrix sizes provide a separate comparison of
the cohomology.  Conditional on the reported numerical vanishings at
larger \(N\), Table~\ref{tab:fortuity} lists the sectors and numbers of
BPS classes identified as fortuitous at fixed fermion number.  The
stable images in the remaining nonzero sectors are not yet determined.

\subsection{Comparison with supersymmetric SYK}

Here the support means the fermion numbers for which \(h_R\ne0\).  It
contains 9 of 17 sectors at \((p,N)=(5,4)\), 11 of 17 at \((7,4)\),
and 10 of 26 at \((5,5)\).

For comparison, consider \(\mathcal N=2\) SYK with the same degree
\(\hat q\) of the supercharge and the same number
\(N_f=N^2=16\) of complex fermions
\cite{FuGaiotto,ChangChenSiaYang2024}.  The examples obtained by exact
diagonalization in those references have BPS states in \(\hat q\)
sectors, and their total count equals the lower bound from the Witten
indices.  The corresponding values are shown below and in
Figure~\ref{fig:charge_support}.

\begin{center}
\renewcommand{\arraystretch}{1.1}
\begin{tabular}{lcccc}
\toprule
Model & \(Z_{\mathrm{BPS}}\) & Index bound & Difference & Sectors \\
\midrule
SYK \(\hat q=3\), \(N_f=16\) & 8{,}748 & 8{,}748 & 0 & 3 \\
SYK \(\hat q=5\), \(N_f=16\) & 38{,}000 & 38{,}000 & 0 & 5 \\
SYK \(\hat q=7\), \(N_f=16\) & 55{,}468 & 55{,}468 & 0 & 7 \\
Matrix \((5,4)\) & 44{,}000 & 38{,}000 & 6{,}000 & 9 \\
Matrix \((7,4)\) & 59{,}136 & 55{,}468 & 3{,}668 & 11 \\
\bottomrule
\end{tabular}
\end{center}

\noindent The comparison uses the lower bound from the Witten indices
in the \(\hat q\) residue classes.
For the two matrix models, the differences are \(6{,}000\) at
\((5,4)\) and \(3{,}668\) at \((7,4)\).
For \(\hat q=7\), the two central SYK multiplicities,
\(h_7=11{,}319\) and \(h_8=12{,}838\)~\cite{ChangChenSiaYang2024},
coincide with the numerical matrix model values at \((7,4)\).
Differences occur away from the center and in the additional outer sectors of the matrix model.

The distribution of supersymmetric ground states across charge sectors
has also been analysed in the double scaling
limit~\cite{BerkoozSYK2020}.  It shapes the phase structure at
nonzero background charge~\cite{HeydemanTuriaciZhao2023}.
Its relation to the Schwarzian description of black holes close to the
BPS bound is developed in
Refs.~\cite{ChangChenSiaYang2024,MuruganStanfordWitten2017,LinMaldacena2023}.

\begin{figure}[H]
\centering
\includegraphics[width=0.72\textwidth]{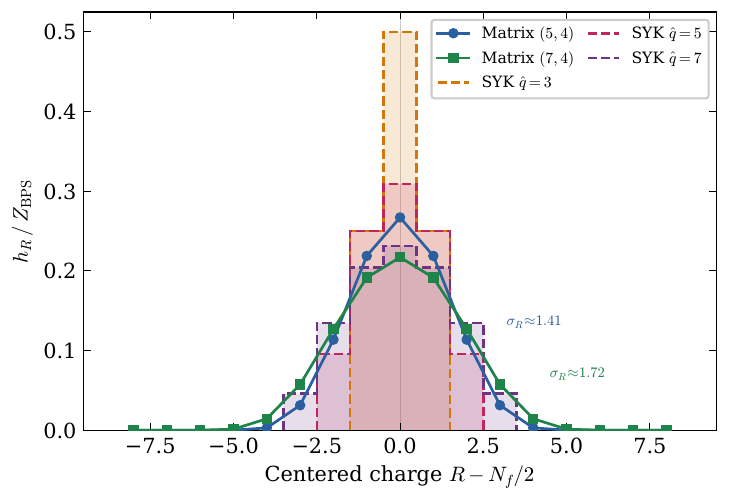}
\caption{Charge support: matrix model vs.\ \(\mathcal N=2\) SYK.
Normalized numerical BPS distributions \(h_R/Z_{\mathrm{BPS}}\) for the matrix
models \((5,4)\) and \((7,4)\) (solid lines, \(N^2=16\)) and the
\(\mathcal N=2\) SYK examples at \(\hat q=3,5,7\) (shaded bars,
\(N_f=16\)).  Every distribution shown concerns a system with 16
complex fermions.
The labels give the standard deviation from half filling,
\(\sigma_R=[\sum_R(R-8)^2h_R/Z_{\mathrm{BPS}}]^{1/2}\) for each matrix
model distribution.
In each computed case, SYK has BPS states in \(\hat q\) sectors and its
BPS count equals the index bound.  The matrix model has nonzero
multiplicities in more sectors and lies strictly above the index bound.
}
\label{fig:charge_support}
\end{figure}

Tensor models without disorder can reproduce the melonic large \(N\)
limit of SYK~\cite{Witten2016SYK,KlebanovTarnopolsky2017}.

Reference~\cite{BiggsMaldacena2026} also studies a model without disorder
in which \(2j+1\) complex fermions transform in the spin \(j\)
representation of \(\SU(2)\).  For observables that are singlets under
\(\SU(2)\), its expansion at large \(j\) is melonic.  At small charge and
low energy, the resulting
equations agree with those of \(\mathcal N=2\) SYK and contain the
corresponding Schwarzian mode.
The present model instead has an exact \(\SU(N)\) symmetry that acts by
conjugation on the matrix fermions.  No melonic limit is constructed
here.  Reference~\cite{ChangChenSiaYang2024} relates the concentration
of BPS states in a narrow range of \(R\) to fortuity in supersymmetric
SYK models.  It also formulates a conjecture concerning quantum chaos
and the supercharge.  The comparison in Figure~\ref{fig:charge_support}
does not by itself determine fortuity under the projections at fixed
\(R\) used here or provide a test of chaos for the present matrix model.

Another comparison is the \(\mathcal N=2\) gauged quantum mechanics
obtained by reducing a Chern--Simons matter theory in three dimensions
on \(S^2\)~\cite{BeniniSoltaniZhang2023}.
It reproduces the entropy of BPS black holes in \(\mathrm{AdS}_4\) and
has statistically distributed couplings, but it arises from string
theory rather than from a supercharge given by a single trace.

\subsection{Bound at large \texorpdfstring{$N$}{N}}
\label{sec:largeN}

Across the three values of \(N\) studied numerically, the proportion of
BPS states decreases while the total BPS count increases, as shown in
Table~\ref{tab:largeN}.

\begin{table}[H]
\centering
\small
\begin{tabular}{crrrc}
\toprule
\(N\) & \(Z_{\mathrm{BPS}}\) & \(2^{N^2}\) & BPS/total & \(\frac{1}{N^2}\log Z_{\mathrm{BPS}}\) \\
\midrule
3 & 440 & 512 & 85.9\% & 0.676 \\
4 & 44{,}000 & 65{,}536 & 67.1\% & 0.668 \\
5 & 15{,}480{,}000 & 33{,}554{,}432 & 46.1\% & 0.662 \\
\bottomrule
\end{tabular}
\caption{Reported numerical quintic BPS data for \(N=3,4,5\).
The BPS fraction decreases across the three displayed values, while the
normalized entropy remains near \(0.66\).  Whether either trend persists
is open.}
\label{tab:largeN}
\end{table}

\noindent The three values do not determine whether
\(N^{-2}\log Z_{\mathrm{BPS}}\) has a limit.  The Witten indices give the
following result for every fixed \(p\).

\begin{theorem}[Bound from the Witten indices]
\label{thm:largeN-bound}
Let \(p\ge3\) be fixed and odd.  The indices in~\eqref{eq:index} obey
\begin{equation}\label{eq:floor-rate}
\lim_{N\to\infty}\frac{1}{N^2}\log\sum_{a=0}^{p-1}|I_a|
=\log\!\bigl(2\cos(\pi/(2p))\bigr).
\end{equation}
Consequently,
\begin{equation}\label{eq:ZBPS-bound}
\liminf_{N\to\infty}\frac{1}{N^2}\log Z^{(p,N)}_{\rm BPS}
\ge\log\!\bigl(2\cos(\pi/(2p))\bigr).
\end{equation}
Moreover, for any sequence of matrix sizes \(N_j\to\infty\) along which
\(N_j^{-2}\log Z^{(p,N_j)}_{\rm BPS}\) converges, its limit lies in
\begin{equation}\label{eq:window}
\Bigl[\log\!\bigl(2\cos(\pi/(2p))\bigr),\log 2\Bigr].
\end{equation}
\end{theorem}

\begin{proof}
Put \(n=N^2\) and \(b_\ell=(1-\omega^\ell)^n\).  Equation~\eqref{eq:index}
is the inverse discrete Fourier transform of \((b_\ell)\).  Expanding its
squared \(\ell^2\)-norm and using the root-of-unity orthogonality relation
\(\sum_{a=0}^{p-1}\omega^{a(\ell-\ell')}=p\,\delta_{\ell,\ell'}\) gives
\begin{equation}\label{eq:index-orthogonality}
\sum_{a=0}^{p-1}|I_a(n)|^2
=\frac{1}{p}\sum_{\ell=0}^{p-1}|1-\omega^\ell|^{2n}.
\end{equation}
For \(\ell=1,\ldots,p-1\), \(|1-\omega^\ell|=2\sin(\pi\ell/p)\).
Since \(p\) is odd, the maximum \(\lambda=2\cos(\pi/(2p))\) is attained
exactly at \(\ell=(p\pm1)/2\).  Therefore, for
\(I(n)=(I_0(n),\ldots,I_{p-1}(n))\),
\[
\sqrt{\frac{2}{p}}\,\lambda^n
\le \lVert I(n)\rVert_2
\le \lVert I(n)\rVert_1
\le \sqrt{p}\,\lVert I(n)\rVert_2
\le \sqrt{p-1}\,\lambda^n.
\]
Taking logarithms, dividing by \(n=N^2\), and letting \(N\to\infty\)
proves~\eqref{eq:floor-rate}.  The signed cohomology identity above gives
\(Z_{\rm BPS}\ge\sum_a|I_a|\), and hence~\eqref{eq:ZBPS-bound}.
Finally, \(Z_{\rm BPS}\le2^{N^2}\), which gives the upper bound
in~\eqref{eq:window}.
\end{proof}

Numerically, the bound is \(0.549\) at \(p=3\), \(0.643\) at \(p=5\), and \(0.668\) at \(p=7\), approaching \(\log 2\approx 0.693\) as \(p\to\infty\).

The observed quintic sequence \(0.676,0.668,0.662\) for \(N=3,4,5\) (Table~\ref{tab:largeN}) lies above~\(0.643\), with a decreasing gap.
This is consistent with \(\frac{1}{N^2}\log Z^{(5,N)}_{\rm BPS}\)
approaching the rate of the index bound from above.
The three available values do not distinguish whether the sequence has
the lower rate in~\eqref{eq:ZBPS-bound} or a larger one.

\subsection{Open questions}

The first open question is whether \((1+x)^N\) divides the BPS
polynomial for every odd \(p\ge5\) with \(p\le 2N-1\).  The cubic case
is known exactly.  Corollary~\ref{cor:structural-factors} proves two
factors for even \(N\) and three factors for odd \(N\).  The additional
factors seen in the data remain unexplained.

Let
\[
\mathcal B^\bullet:=\bigoplus_R\mathcal B_R^{(p,N)},
\]
and let \(\mathcal B_0^\bullet\) denote the cohomology of the traceless
complex.  Theorem~\ref{thm:U1} and the K\"{u}nneth formula give an
isomorphism of graded vector spaces
\[
\mathcal B^\bullet\cong
\Lambda^\bullet(\C\chi)\otimes\mathcal B_0^\bullet.
\]
Each traceless linear combination of the diagonal creation operators
anticommutes with \(Q_p\) and therefore induces a map of degree one on
\(\mathcal B_0^\bullet\).  The traceless diagonal combinations form a
vector space of dimension \(N-1\).  Their induced maps anticommute and
square to zero, so they make \(\mathcal B_0^\bullet\) a graded module over
\(\Lambda^\bullet(\C^{N-1})\).

A sufficient explanation for the remaining factors of \(1+x\) would be
the existence of a graded vector space \(\mathcal M^{(p,N)}\) and an
isomorphism of graded modules
\[
\mathcal B_0^\bullet
\cong\Lambda^\bullet(\C^{N-1})\otimes \mathcal M^{(p,N)}.
\]
Write \(\operatorname{Hilb}\mathcal M^{(p,N)}(x)\) for the Hilbert series
of \(\mathcal M^{(p,N)}\).  Taking Hilbert series would then give
\[
\Zbps^{(p,N)}(x)
=(1+x)^N\operatorname{Hilb}\mathcal M^{(p,N)}(x).
\]
The quotient by \((1+x)^N\) would therefore have nonnegative integer
coefficients.
If the full factorization in~\eqref{eq:factorlaw} holds, then
\[
\operatorname{Hilb}\mathcal M^{(p,N)}(x)
=p^{b_{p,N}}x^{q_{\min}}T_N^{(p)}(x).
\]
The freeness statement would not by itself explain the common power of
\(p\).  Such freeness is not implied by divisibility of the Hilbert
series and is not proved here.
For broader context, Refs.~\cite{CremonesiHananyZaffaroni2014,
CremonesiHananyMekareeya2014} develop Hilbert series for Coulomb branches
of three-dimensional \(\mathcal N=4\) theories.

The second question concerns the exponent \(b_{p,N}\).  The sequence
\(b_{5,N}=1,3,4\) for \(N=3,4,5\) has no known explicit formula.
At \(N=5\), the reported \(5^4\) divisibility of \(h_R\) does not arise
from \(5^4\) divisibility of \(\dim\ker Q_R\) and
\(\operatorname{rank}Q_{R-5}\) separately.  In the first two surviving
sectors \(R=8,9\), the nullities and incoming ranks are divisible by
\(5^2\), while at \(R=10\) they are not.  Nevertheless, the differences
\(h_R\) are all divisible by \(5^4\).  For the reported numerical ranks,
the relation is
\[
\dim\ker Q_R\equiv \mathrm{rank}\,Q_{R-5}\pmod{5^4},
\]
although neither term on the two sides is always divisible by \(5^4\).
Explaining this congruence may help determine \(b_{p,N}\).

The third question is the behaviour as \(N\) grows.  At fixed \(p\),
Theorem~\ref{thm:largeN-bound} gives the interval~\eqref{eq:window}, but
the available values of \(N\) do not determine the limiting rate.  A
different problem is the joint limit along odd integers \(p=p(N)\) that
satisfy \(p(N)\le2N-1\) and \(p(N)/N\to\alpha\in(0,2]\).  The present
calculation does not determine the asymptotic behaviour of
\(N^{-2}\log\Zbps^{(p(N),N)}(1)\) in this regime.

A separate question is whether the model admits, for every \(N\ge3\), a
family of pairwise commuting conserved operators that are
invariant under \(\SU(N)\), go beyond functions of \(H\), and resolve
equivalent irreducible summands.  The nonzero commutators among the
displayed \(K_k\) show that these operators cannot form such a family in
the stated \(N=4\) sectors.  They do not prove chaos or a failure of
integrability, and they do not exclude a different family.

Finally, the stable images at fixed fermion number remain undetermined
for \(R=8,\ldots,12\) at \((p,N)=(5,4)\) and for
\(R=6,\ldots,13\) at \((7,4)\).  No stable image of a nonzero BPS sector
at \(N=5\) is determined by the present data, because that would require
cohomology at larger \(N\).  At rank \(N\), particle--hole conjugation
sends sector \(R\) to sector \(N^2-R\).  Across different values of \(N\),
the latter sectors have a fixed number of holes rather than a fixed
fermion number, so conjugation does not determine the stable images at
fixed \(R\).

\section*{Acknowledgments}

The author thanks Hank Chen, Yiming Chen, Leonardo Santilli, Jan Troost
and Alberto Zaffaroni for helpful comments and correspondence.

\appendix
\renewcommand{\thesection}{Appendix~\Alph{section}}
\renewcommand{\thesubsection}{\Alph{section}.\arabic{subsection}}

\section{Deferred proofs and algebraic background}
\label{app:math}

This appendix supplies the proofs deferred from the main text and
summarizes the facts about cohomology of Lie algebras used in the model.

\subsection{Projections between ranks and fortuity}
\label{app:projections}

\begin{proof}[Proof of Proposition~\ref{thm:rankproj}]
Since $\pi_{M\to N}$ is defined as an algebra homomorphism on $\mathcal H_M$
(sending each $\Psi_{ij}$ to itself if $i,j\le N$ and to zero otherwise,
and extending multiplicatively), write
$q_p^{(N)}:=\sum_{i_1,\ldots,i_p=1}^{N}\Psi_{i_1 i_2}\cdots\Psi_{i_p i_1}$
so that $Q_p^{(N)}(v)=q_p^{(N)}\wedge v$.
Then $\pi_{M\to N}(q_p^{(M)})=q_p^{(N)}$,
since the projection sets to zero every term in $q_p^{(M)}$ involving an index $>N$.
Therefore for all $v\in\mathcal H_M$:
\[
\begin{aligned}
\pi_{M\to N}(Q_p^{(M)}v)
&=\pi_{M\to N}(q_p^{(M)}\wedge v)\\
&=\pi_{M\to N}(q_p^{(M)})\wedge\pi_{M\to N}(v)\\
&=q_p^{(N)}\wedge\pi_{M\to N}(v)
=Q_p^{(N)}\pi_{M\to N}(v),
\end{aligned}
\]
proving the intertwining identity for all $M>N$.
The intertwining identity shows that $\pi_{M\to N}$ sends
$\ker Q_p^{(M)}$ into
$\ker Q_p^{(N)}$ and $\operatorname{im}Q_p^{(M)}$ into
$\operatorname{im}Q_p^{(N)}$, so the induced maps on cohomology are
well defined.  Transitivity follows from
$\pi_{L\to N}=\pi_{M\to N}\circ\pi_{L\to M}$.
The images $\operatorname{Im}\Pi_{M\to N}^{(R)}$ form a descending chain
in the space $\mathcal B_R^{(p,N)}$, which has finite dimension.
For $L>M>N$, transitivity gives
$\operatorname{Im}\Pi_{L\to N}^{(R)}\subseteq\operatorname{Im}\Pi_{M\to N}^{(R)}$.
The finite dimension ensures that the chain stabilizes,
so $\mathcal B_{R,\mathrm{mon}}^{(p,N)}$ equals the eventual stable image.
\end{proof}

\paragraph{Relation to the covering definition.}
Reference~\cite{Chang2024} defines the monotone subspace as the image in
cohomology of a map from one covering complex.  The present paper instead
uses the eventual image of the maps \(\Pi_{M\to N}^{(R)}\) as \(M\)
increases.  Proposition~\ref{thm:rankproj} shows that these images are
well defined in cohomology and stabilize at each fixed fermion number.
This is analogous to the image used in Ref.~\cite{Chang2024}, but no
equivalence of the two constructions is assumed.  In the D1--D5
system, the analogous projection need not commute with the
supercharge~\cite{ChangLinZhang2025}.  Here commutation follows directly
because \(Q_p=\tr(\Psi^p)\) contains only creation operators.

\subsection{Hodge duality at \texorpdfstring{$(p,N)=(5,3)$}{(p,N)=(5,3)}}
\label{app:hodge-duality}

\begin{proof}[Proof of Proposition~\ref{prop:su3-commutativity}]
Throughout this proof a superscript records the model to which an
object of Section~\ref{sec:normalH} belongs.  The map
\(\mathcal A^{(p)}_k\) is defined in \eqref{eq:contraction-map} using
\(\Omega_p\), and \(K^{(p,N)}_k\) is the operator in
\eqref{eq:Kk-definition} for
the model with parameters \((p,N)\).
Here \(N=3\) throughout, and \(K_k\) means \(K^{(5,3)}_k\).

Work first on the traceless Fock space
\(\Lambda^\bullet W_0\), where \(W_0\simeq\mathfrak{sl}_3\).
Fix an orthonormal basis \(\{T_a\}_{a=1}^{8}\) of traceless Hermitian
matrices and let \(W_{0,\mathbb R}\) be its real span.  Thus
\(W_0=W_{0,\mathbb R}\otimes_{\mathbb R}\C\), and the adjoint action of
\(\SU(3)\) on \(W_{0,\mathbb R}\) lies in \(SO(8)\).
We use the complex linear Hodge star obtained by extending the real
Euclidean Hodge star on \(W_{0,\mathbb R}\).  In the basis \(T_a\) it is
a signed permutation of basis wedges, so it is unitary for the induced
Hermitian inner product and commutes with the \(\SU(3)\) action.
The relevant dimension equality is
\[
\dim_\C W_0=N^2-1=8=3+5 ,
\]
special to \(N=3\), which makes \({*}\) carry \(\Lambda^3W_0^*\)
isomorphically onto \(\Lambda^5W_0^*\).
The spaces of invariant forms of degrees three and five on \(W_0\)
both have dimension one~\cite{ChevalleyEilenberg,Koszul1950}.
Therefore, by equivariance,
\begin{equation}\label{eq:su3-hodge-dual}
\Omega_5=c\,{*}\Omega_3
\end{equation}
for some nonzero complex scalar \(c\).
Only the modulus of~\(c\) is needed below.
In the convention fixed by \(\widetilde Q_p=p^{-1}\tr(A^p)\),
\eqref{eq:K0p5} and \eqref{eq:cubic-components} give
\(\|\Omega_5\|^2=24\) and \(\|\Omega_3\|^2=8\).  Since \({*}\) is
unitary, \eqref{eq:su3-hodge-dual} forces
\(|c|^2=3\).
This is the \(\SU(3)\) duality between the primitive
cubic and quintic cocycles
~\cite{deAzcarragaMacfarlane2000,deAzcarragaMacfarlaneCocycles}.

Set \(q=\widetilde Q_3\), so that \(q\) acts on \(\Lambda^\bullet W_0\)
as exterior multiplication by \(\Omega_3\), and put \(B=q^\dagger q\).
We claim that, on \(\Lambda^kW_0\) for \(0\le k\le5\),
\begin{equation}\label{eq:su3-kernel-lift}
\bigl(\mathcal A_k^{(5)}\bigr)^\dagger\mathcal A_k^{(5)}
=3\,B\big|_{\Lambda^kW_0}.
\end{equation}
We use the standard relation between contraction of a Hodge dual and
exterior multiplication.
Write \(\flat\) for the complex linear isomorphism
\(\Lambda^kW_0\to\Lambda^kW_0^*\) obtained from the complex bilinear
extension of the real Euclidean form for which the \(T_a\) are
orthonormal.
Then, for \(\alpha\in\Lambda^kW_0\) and \(\omega\in\Lambda^mW_0^*\),
\begin{equation}\label{eq:hodge-contraction}
\iota_\alpha({*}\omega)=\varepsilon\,{*}(\omega\wedge\alpha^\flat),
\end{equation}
where \(\varepsilon\in\{\pm1\}\) depends only on \(k\), \(m\) and
\(\dim W_0\), not on \(\alpha\) or \(\omega\).
For \(k=1\) this is the usual relation
\(\iota_v({*}\omega)={*}(\omega\wedge v^\flat)\)~\cite{GriffithsHarris1978}.
The general case follows by iterating it, each step contributing a sign
that depends only on the degrees involved.
We suppress \(\flat\) from here on and use this complex linear
identification of \(W_0\) with \(W_0^*\).
Applying \eqref{eq:hodge-contraction} with \(\omega=\Omega_3\) and using
\eqref{eq:su3-hodge-dual}, we get, for all
\(\alpha,\beta\in\Lambda^kW_0\),
\[
\bigl\langle
\mathcal A_k^{(5)}\alpha,\ \mathcal A_k^{(5)}\beta
\bigr\rangle
=|c|^2\bigl\langle
{*}(\Omega_3\wedge\alpha),\ {*}(\Omega_3\wedge\beta)
\bigr\rangle
=3\,\langle \Omega_3\wedge\alpha,\ \Omega_3\wedge\beta\rangle
=3\,\langle q^\dagger q\,\alpha,\ \beta\rangle .
\]
The sign \(\varepsilon\) cancels because it occurs in both arguments of
the sesquilinear form.  The middle equality uses the unitarity of
\({*}\).
This is \eqref{eq:su3-kernel-lift}.
This identity, including the factor~\(3\), is independently verified by
exact arithmetic for every \(k\) in the ancillary program
\texttt{verify\_normal\_ordering.py}.

After normal ordering, let \(B_{[j]}\) denote the part of \(B\) containing
\(j\) creation and \(j\) annihilation operators.  Then
\[
B=\sum_{j=0}^{3}B_{[j]} .
\]
Its components follow from the cubic identities
\eqref{eq:cubic-components} and from the cancellation of the
term with three creation and three annihilation operators in
\(\{q,q^\dagger\}\):
\begin{equation}\label{eq:su3-cubic-bodies}
B_{[0]}=8,\qquad
B_{[1]}=-3\hat n,\qquad
B_{[2]}=3\hat n-\hat C_2,\qquad
B_{[3]}=-q q^\dagger .
\end{equation}

It remains to lift \eqref{eq:su3-kernel-lift} from the sector with
\(k\) particles to the full Fock space.  By
\eqref{eq:Kk-definition}, \(K_k\) preserves fermion number and its
matrix on the sector with \(k\) particles is
\(\bigl(\mathcal A_k^{(5)}\bigr)^\dagger\mathcal A_k^{(5)}\).
For \(|I|=|J|=k\), let \(|I\rangle=a_I^\dagger|0\rangle\) and define
\(B_{I,J}:=\langle I|B|J\rangle\).  By \eqref{eq:su3-kernel-lift}, the
\((I,J)\) entry of this matrix is \(3B_{I,J}\).
More explicitly, consider a term \(a_{I'}^\dagger a_{J'}\) in
\(B_{[j]}\) written in normal order, with
\(|I'|=|J'|=j\).  Its contribution to the \(k\)-particle matrix is
obtained by adjoining the same \((k-j)\)-element set \(T\), disjoint
from \(I'\cup J'\), to the input and output index sets.  The signs
required to order \(I'\cup T\) and \(J'\cup T\) occur once in the
matrix element and once in the corresponding operator in normal order,
so they cancel.  Let \(|O\rangle\) be a basis state with \(r\) occupied
modes.  If \(a_{I'}^\dagger a_{J'}|O\rangle\ne0\), then
\(J'\subset O\) and \(I'\cap(O\setminus J')=\varnothing\).  The allowed
sets \(T\) are therefore exactly the \((k-j)\)-element subsets of
\(O\setminus J'\), of which there are \(\binom{r-j}{k-j}\).  If the
original term annihilates \(|O\rangle\), every term obtained by
adjoining \(T\) also vanishes.  The lift of \(B_{[j]}\) is therefore
\(\binom{\hat n-j}{k-j}B_{[j]}\).
Consequently
\begin{equation}\label{eq:su3-binomial-lift}
K_k^{(5,3)}
=3\sum_{j=0}^{\min(3,k)}
\binom{\hat n-j}{k-j}B_{[j]},
\qquad 0\le k\le4 ,
\end{equation}
as an identity of operators on each sector with fixed \(\hat n\).
This identity is also checked exactly in the ancillary program.
The four operators in \eqref{eq:su3-cubic-bodies} commute:
the first three are functions of \(\hat n\) and \(\hat C_2\), while
\(q q^\dagger\) preserves \(\hat n\) and is \(\SU(3)\)-invariant, hence
commutes with both.
Since \(\hat n\) commutes with all of them, every \(K_k^{(5,3)}\) in
\eqref{eq:su3-binomial-lift} is a polynomial in the four commuting
operators \(\hat n\), \(\hat C_2\), \(qq^\dagger\) and the identity.
Equation~\eqref{eq:su3-binomial-lift} therefore proves that the five
operators \(K_0,\ldots,K_4\) commute pairwise.
Finally, on the full Fock space they act as
\(K_k|_{\Lambda^\bullet W_0}\otimes\mathbf 1_{\Lambda^\bullet(\chi)}\),
so the result extends to the trace factor.
\end{proof}

\subsection{Cohomological background}
\label{app:cohom-background}

\paragraph{Relation to standard cohomological complexes.}
The BPS space is the cohomology of
\[
\bigl(\Lambda^\bullet(\mathfrak{gl}_N),\,Q_p\wedge\bigr).
\]
The differential has degree~\(p\), so fixing fermion number modulo~\(p\)
gives \(p\) independent complexes.  After the exact \(U(1)\)
decoupling of Section~\ref{sec:U1}, the complex is
\[
\Lambda^\bullet(\chi)\otimes
\bigl(\Lambda^\bullet(\mathfrak{sl}_N),\,Q_p\wedge\bigr).
\]
This is not the Chevalley--Eilenberg complex~\cite{ChevalleyEilenberg}.
That differential includes contraction.  For odd \(p\ge3\), the
differential here is exterior multiplication by a primitive cocycle of
degree \(p\) when \(p\le2N-1\), and it vanishes when \(p>2N-1\).  In the
Aomoto complex, exterior multiplication by a class of degree one acts
on an Orlik--Solomon algebra~\cite{Aomoto1975,OrlikTerao2001}.  Here the
underlying algebra is the full exterior algebra and the multiplying
class has degree~\(p\).

What is classical is the existence of the primitive generators and the structure of the invariant algebra of \(\Lambda^\bullet\mathfrak g\)~\cite{ChevalleyEilenberg,Koszul1950,Itoh2015}.
Here we study this cohomology at each fermion number beyond the cubic
case.  For \(p=3\), Chen's Hamiltonian is a constant minus the
quadratic Casimir~\cite{Chen2025}.  Kostant's results describe the
relevant Casimir spectrum on the exterior algebra~\cite{Kostant1965}.
For odd \(p\ge5\) with \(p\le2N-1\), the Hamiltonian need
not be determined by the quadratic Casimir, and the cohomology can be
computed by linear algebra.  The
highest nonzero generator of the form \(\tr(A^p)\) occurs at
\(p=2N-1\).

\paragraph{Generators and the condition \(p\le2N-1\).}
Let \(V=\C^N\), and let \(X=(X_{ij})\) be the generic anticommuting
\(N\times N\) matrix whose entries generate
\(\Lambda(V\otimes V^*)\).  With \(GL(V)\) acting by conjugation, the
invariant subalgebra \(\Lambda(V\otimes V^*)^{GL(V)}\) is generated by
\[
\tr X,\quad \tr X^3,\quad\ldots,\quad \tr X^{2N-1},
\]
with no relations beyond anticommutativity~\cite{Itoh2015}.
For the traceless matrix \(A\) defined in
\eqref{eq:trace-decomposition}, the generator \(\tr X\) of degree one is
absent.
The anticommuting Cayley--Hamilton identity of
Ref.~\cite{Itoh2015} then implies \(\tr(A^p)=0\) for odd
\(p>2N-1\).  Indeed, multiplying that identity by \(A^{2r}\), taking
the trace, and using the vanishing of even traces gives
\[
N\tr(A^{2N-1+2r})=0,\qquad r\ge1.
\]
This algebraic result has a cohomological counterpart.
The Chevalley--Eilenberg theorem~\cite{ChevalleyEilenberg} gives
\[
H^*(\mathfrak{su}(N);\C)\cong \Lambda(c_3,c_5,\ldots,c_{2N-1}),
\]
where \(c_{2m+1}\), for \(m=1,\ldots,N-1\), is the primitive
cohomology class associated with the exponent \(m\) of
\(\mathfrak{su}(N)\)~\cite{ChevalleyEilenberg,Koszul1950}.
Using the trace form to identify \(\mathfrak{sl}_N\) with its dual,
the class \(c_{2m+1}\) is represented, up to normalization, by
\(\tr(A^{2m+1})\) in the invariant exterior algebra
~\cite{Itoh2015,Troost2020}.  Under the standard transgression, the
same class corresponds to the invariant polynomial
\(Y\mapsto\tr(Y^{m+1})\), where \(Y\in\mathfrak{sl}_N\) is an ordinary
commuting matrix variable~\cite{Koszul1950}.  Thus the cohomological
description gives the nonzero traces in degrees
\(3,5,\ldots,2N-1\), while the anticommuting Cayley--Hamilton identity
gives their vanishing above this range.  Therefore, for odd
\(p\ge3\), \(\tr(A^p)\) is nonzero exactly when
\begin{equation}\label{eq:range}
p\le 2N-1.
\end{equation}

\section{Data for each fermion number}
\label{app:data}

The spectral computations at \(N=3\) and \(N=4\) used floating-point
arithmetic with double precision and a decomposition by singular values
of the stacked matrix \(A_R\) in
\eqref{eq:stackedA}.
Fock states are ordered lexicographically by bit pattern, and the
\(Q_p\) matrices are built by explicit enumeration of all closed index
chains in \(\tr(\Psi^p)\), equivalently \(\tr(A^p)\) by
Theorem~\ref{thm:U1}.
Ranks were determined with a threshold of \(10^{-10}\) relative to the
largest singular value.  At \(N=3,4\), the zero and nonzero clusters are
separated by more than four orders of magnitude.

The \(N=5\) computations used sparse matrix assembly, with inner loops
accelerated by Numba, to build each sector map \(Q_R\) as a compressed
sparse row matrix.
The Gram matrix \(Q_R^\top Q_R\) was then decomposed into connected components via its sparsity graph.
Each block was diagonalized independently using dense \texttt{eigvalsh} or LDL\(^\top\) factorization.
In every block, the zero and nonzero eigenvalue clusters are separated by more than two orders of magnitude.
A second implementation reproduced all reported \(N=5\) ranks.  It
finds the connected components by a disjoint set algorithm and counts
the LDL\(^\top\) pivots in each block.

The complete rank arrays for the four cases computed in every sector are
supplied in the ancillary file \texttt{rank\_data.json}.  For \((5,3)\),
\((5,4)\), and \((7,4)\), \texttt{verify\_ranks.py} reconstructs the
sector maps and performs the checks over finite fields described in
Section~\ref{sec:gram}.  For \((5,5)\), the entries through \(R=10\) record
the numerical block computations described above.  The remaining
nontrivial ranks follow from the symmetry in
Eq.~\eqref{eq:rank-palindrome}.  The script does not
attempt the \(N=5\) case.

\subsection{Energy levels: quintic model, \texorpdfstring{$p=5$}{p=5}, \texorpdfstring{$N=4$}{N=4}}
\label{app:data54}

The table gives the spectrum through half filling, \(0\le R\le8\).
Blank cells denote zero multiplicity.

\begin{center}
\small
\begin{tabular}{c|rrrrrrrrr|r}
\toprule
\(R\) & \(E=0\) & 225 & 600 & 900 & 1100 & 1400 & 2100 & 2400 & 3600 & dim \\
\midrule
0 & 0 &  &  &  &  &  &  &  & 1 & 1 \\
1 & 0 &  &  &  &  &  &  & 15 & 1 & 16 \\
2 & 0 &  &  &  &  & 90 &  & 30 &  & 120 \\
3 & 0 &  & 245 &  & 84 & 195 & 20 & 15 & 1 & 560 \\
4 & 125 & 512 & 574 & 175 & 258 & 120 & 55 &  & 1 & 1820 \\
5 & 1375 & 1536 & 497 & 525 & 354 & 15 & 65 &  & 1 & 4368 \\
6 & 5000 & 1536 & 497 & 525 & 354 & 15 & 65 & 15 & 1 & 8008 \\
7 & 9625 & 512 & 574 & 175 & 258 & 210 & 55 & 30 & 1 & 11440 \\
8 & 11750 &  & 490 &  & 168 & 390 & 40 & 30 & 2 & 12870 \\
\bottomrule
\end{tabular}
\end{center}

Sectors \(R=9,\dots,16\) follow by particle--hole conjugation.
The total BPS count is \(44{,}000\).

\subsection{Other computed energy levels}

\paragraph{Quintic model at \(N=3\).}
The distinct energies are \(\{0,225,600\}\), and the BPS multiplicities are given in \eqref{eq:h53}.

\paragraph{Model of degree seven at \(N=4\).}
The distinct energies are
\[
\{0,1764,4410,7350,7840,14700,18816,35280\}.
\]
The BPS multiplicities are given in \eqref{eq:h74}.
The total BPS count is \(59{,}136\).

\begin{sloppypar}
\hbadness=3000

\end{sloppypar}

\end{document}